\documentclass[10pt]{article}
\usepackage{graphicx}
\usepackage[T1]{fontenc}
\usepackage[utf8]{inputenc}
\usepackage{amssymb}
\usepackage{amsfonts}
\usepackage{dsfont}
\usepackage{mathtools}
\usepackage{amsthm}
\usepackage{amsmath}
\usepackage{textcomp}
\usepackage{eurosym}
\usepackage{stmaryrd}
\usepackage{xcolor}
\usepackage[multiple]{footmisc}

\usepackage{bigints}
\usepackage{geometry}
\newtheorem{lemme}{Lemma}
\newtheorem{thm}{Theorem}
\newtheorem{prop}{Proposition}
\newtheorem{defi}{Definition}
\newtheorem{rem}{Remark}

\begin{document}

\title{Algorithmic market making in dealer markets with hedging and market impact}

\author{Alexander \textsc{Barzykin}\footnote{HSBC, 8 Canada Square, Canary Wharf, London E14 5HQ, United Kingdom, alexander.barzykin@hsbc.com} \and Philippe \textsc{Bergault}\footnote{Ecole Polytechnique, CMAP, 91128 Palaiseau, France, philippe.bergault@polytechnique.edu} \and Olivier \textsc{Guéant}\footnote{Université Paris 1 Panthéon-Sorbonne, Centre d'Economie de la Sorbonne, 106 Boulevard de l'Hôpital, 75642 Paris Cedex 13, France, olivier.gueant@univ-paris1.fr}}
\date{}

\maketitle

\begin{abstract}
In dealer markets, dealers provide prices at which they agree to buy and sell the assets and securities they have in their scope. With ever increasing trading volume, this quoting task has to be done algorithmically in most markets such as foreign exchange markets or corporate bond markets. Over the last ten years, many mathematical models have been designed that can be the basis of quoting algorithms in dealer markets. Nevertheless, in most (if not all) models, the dealer is a pure internalizer, setting quotes and waiting for clients. However, on many dealer markets, dealers also have access to an inter-dealer market  or even public trading venues where they can hedge part of their inventory. In this paper, we propose a model taking this possibility into account, therefore allowing dealers to externalize part of their risk. The model displays an important feature well known to practitioners that within a certain inventory range the dealer internalizes the flow by appropriately adjusting the quotes and starts externalizing outside of that range. The larger the franchise, the wider is the inventory range suitable for pure internalization. The model is illustrated numerically with realistic parameters for USDCNH spot market.     

\medskip
\noindent{\bf Key words:} Dealer markets, Market making, Algorithmic trading, Stochastic optimal control, Viscosity solutions \vspace{5mm}
\end{abstract}

\setlength\parindent{0pt}

\section{Introduction}

In financial markets, liquidity has traditionally been provided by a specific category of agents who, on a continuous and regular basis, set prices at which they agree to buy or sell assets and securities. These agents, called market makers or dealers, play a key role in the price formation process in all markets, but their exact role and behavior depend on the considered asset class.\\

In most order-driven markets, such as stock markets, traditional exchanges have converted from open outcry communications between human traders to electronic platforms organized around all-to-all limit order books, and computers now handle almost all market activity. Official market makers and traditional market making companies still make money by providing liquidity to the market but they are now, somehow, in competition with all market participants who can post liquidity-providing orders. In dealer markets or quote-driven markets, electronification has also been one of the major upheavals of the last decade, with important consequences for dealers. In foreign exchange (FX) cash markets for instance, dealers set up their own private electronic platforms enabling clients to be connected to their stream (clients use a request for stream -- RFS) and to directly send them requests for quotes (RFQ). Similarly, on corporate bond markets, many single-dealer-to-client and multi-dealer-to-client platforms have emerged allowing clients to send RFQs to dealers. Some dealer-to-dealer and all-to-all platforms also exist in these markets, therefore blurring the frontier between OTC and organized markets (see~\cite{schrimpf2019fx, schrimpf2019sizing} for a recent analysis of FX markets).\\

Alongside the multifaceted electronification of trading means, most human market making and quoting activities have been replaced by algorithms. This evolution has naturally gone along with the development of many mathematical models in the academic literature. In 2008, inspired by a paper from the 1980s by Ho and Stoll~\cite{ho1981optimal}, Avellaneda and Stoikov~\cite{avellaneda2008high} proposed a stochastic optimal control model to determine the optimal bid and ask quotes that a single-asset risk-averse market maker should set. The authors paved the way to a new literature on market making that complements the contributions of the economic literature on the topic.\footnote{Economists had studied for a long time the behaviour of dealers with the aim of understanding market liquidity and the magnitude of bid-ask spreads. Models where one or several risk-averse market makers optimize their pricing policy for managing their inventory risk models include Amihud and Mendelson~\cite{amihud1980dealership}, Ho and Stoll~\cite{ho1981optimal, ho1983dynamics}, and O'Hara and Oldfield~\cite{o1986microeconomics}. Models focused on information asymmetries where bid-ask spreads derive from adverse selection include Copeland and Galai~\cite{copeland1983information} and Glosten and Milgrom~\cite{glosten1985bid}. Other classic economic references on market making include Grossman and Miller~\cite{grossman1988liquidity} and the review paper of Stoll~\cite{stoll2003market}.}\\

To build a relevant market making model for order-driven markets, and especially stock markets, it is important to take microstructure into account. For instance, Guilbaud and Pham~\cite{guilbaud2013optimal} modeled the market bid-ask spread with a discrete Markov chain and studied the performance of a market maker submitting limit buy and sell orders at the best limits or in front of them. Guilbaud and Pham~\cite{guilbaud2015optimal} studied a similar problem of optimal market making in a pro-rata limit order book, where the market maker may post limit orders but also market orders represented by impulse controls. Fodra and Pham~\cite{fodra2015high, fodra2015semi} considered a model in which market orders arrive in the limit order book according to a point process correlated with the stock price itself. They modeled a market maker as an agent placing limit orders of constant size at the best bid and at the best ask, and solved the problem faced by a risk-averse market maker. More recently, in Abergel et al.~\cite{abergel2020algorithmic}, the authors proposed a different model for the limit order book (first introduced in Abergel et al.~\cite{abergel2016limit}), in which the limit orders, market orders, and cancel orders arrive according to Markov jump processes with intensities depending only on the state of the limit order book. They considered the case of a market maker trading in this limit order book, and proposed a quantization-based algorithm to numerically solve the resulting high-dimensional problem. Finally, Capponi et al.~\cite{capponi2021market} studied a discrete-time problem, assuming that the market maker can place bid and ask limit orders simultaneously on both sides at prespecified dates. In their framework, the number of filled orders during each period depends linearly on the distance between the fundamental price and the price of the market maker's limit order, with random slope and intercept coefficients. Using discrete-time optimal control theory, the authors managed to get an explicit characterization of the optimal strategy.\\

In the case of dealer markets,\footnote{In fact, the models may be relevant for order-driven markets in which the spread-to-tick ratio is large.} most market making models derive from the seminal work of Avellaneda and Stoikov~\cite{avellaneda2008high}.\footnote{See the books of Cartea et al.~\cite{cartea2015algorithmic} and  Guéant~\cite{gueant2016financial} for a detailed discussion.} For instance, Guéant et al.~\cite{gueant2013dealing} provided a rigorous analysis of the stochastic optimal control problem introduced in Avellaneda and Stoikov~\cite{avellaneda2008high} and showed that, by adding risk limits to the inventory of the market maker, the problem boils down to a finite system of ordinary differential equations (ODEs) -- in particular, those ODEs are linear in the case of the exponential intensity functions proposed in~\cite{avellaneda2008high}. Cartea, Jaimungal, and Ricci considered in~\cite{cartea2014buy} a different kind of objective function: instead of the Von Neumann-Morgenstern expected utility of~\cite{avellaneda2008high}, they optimized a risk-adjusted expectation of the PnL. Cartea and Jaimungal, along with diverse coauthors, then proposed several extensions. For instance, Cartea, Donnelly, and Jaimungal studied in~\cite{cartea2017algorithmic} the impact of uncertainty on the parameters of the model. Multi-asset extensions have been proposed by Guéant in~\cite{gueant2016financial, gueant2017optimal} for both kinds of objective function, and the author showed again that the problem boils down to a system of ODEs. In all these models the trade size is assumed constant: the same number of assets is bought or sold at each trade. In~\cite{bergault2019size}, Bergault and Guéant introduced a distribution of trade sizes in the model along with the possibility for the dealer to answer different quotes for different sizes. They characterized the value function of the problem with an integro-differential equation of the Hamilton-Jacobi~(HJ) type that can be tackled without viscosity solution but using functional analysis and ODE techniques (in infinite dimension). They also proposed a dimensionality reduction technique in order to approximate numerically the optimal bid and ask quotes  in high dimension, by projecting the market risk on a low-dimensional space of factors. This problem of approximating the solution across a large universe of assets has been addressed using different approaches. In Bergault et al.~\cite{evangelista2020closed}, the authors regarded the Hamilton-Jacobi equation of the problem as a perturbation of another simpler equation whose solution can be computed in closed form. In Guéant and Manziuk~\cite{gueant2019deep}, the authors used neural networks instead of grids to compute an approximate solution -- a method inspired by approximate dynamic programming and reinforcement learning techniques.\\

In most market making models adapted to dealer markets, the dealers (who act as market makers) are pure liquidity providers: they buy and sell assets at the bid and ask prices they quote. Of course, they seldom buy and sell simultaneously: they carry inventory and bear price risk. The problem faced by dealers in these models is a subtle dynamic optimization problem in which they must mitigate the risk associated with price changes by skewing their quotes as a function of their inventory. In practice, dealers in most dealer markets have an additional way to manage their inventory risk since they can partially or completely hedge it by trading on the Dealer-to-Dealer (D2D) segment of the market and in a variety of all-to-all platforms.\footnote{Regarding the trading mechanisms on these platforms, they are often based on limit order books but we do not model order books and rather consider a modeling framework inspired from the literature on optimal execution (see below).} To be generic, we say throughout this paper that the dealer has access to liquidity pools.\\ 

The co-existence of requests for quotes and requests for stream on the one hand and access to liquidity pools on the other hand has seldom been studied in the academic literature on optimal market making (the only instance we found beyond our paper is the very recent paper~\cite{bakshaev2020market} that proposes a reinforcement learning approach). The trade-off between internalization and externalization is nevertheless discussed in the literature. It is discussed by Butz and Oomen in~\cite{butz2019internalisation} on the basis of queuing theory, though not with optimized quotes. It is also abundantly discussed on empirical grounds in the recent BIS Triennial Survey that concludes on the growing prevalence of internalization in FX markets (see~\cite{schrimpf2019fx}). A wide spectrum of behaviors is documented in~\cite{schrimpf2019fx}, from pure externalization to large ratios of internalization. It is noteworthy that even though internalization ratios for top trading centers exceed 80\% in FX markets, hedging through externalization remains an essential component of risk management.\\

The main goal of our paper is to build an optimal strategy for the dealer that includes the possibility to hedge by buying and selling (in continuous time) in a liquidity pool, in order to better mitigate inventory risk. By trading in a liquidity pool, the dealer adds certainty to inventory risk management but that comes with execution costs and market impact, in part due to the communication of trading intentions to a wider audience. Our setup is inspired by Almgren-Chriss-like models of optimal execution (see Almgren et al.~\cite{almgren2003optimal, almgren2001optimal}, and Guéant~\cite{gueant2016financial} for a general presentation).\footnote{Our model can be seen as an encounter between two topics: optimal market making and optimal execution, i.e. Avellaneda-Stoikov and Almgren-Chriss.} More precisely, compared to existing market making models, ours includes a new form of control -- in addition to the bid and ask quotes -- that represents the execution rate of the dealer in a liquidity pool and features (i) execution costs to proxy transaction costs and nonlinear liquidity costs, and (ii) permanent market impact (assumed to be linear in the execution rate).\\

In Section~\ref{model}, we present our  model involving an asset for which the dealer has a classical quoting activity together with the possibility to hedge risk by trading in a liquidity pool. We then introduce the stochastic optimal control problem of the dealer. In Section~\ref{viscoSec}, we characterize the associated value function as the unique continuous viscosity solution of a Hamilton-Jacobi equation. It is noteworthy that the Hamilton-Jacobi equation associated with our problem is of a new type compared to those associated with existing models: it is a partial integro-differential equation (PIDE) that involves finite differences terms (in the inventory variable) like in the infinite-dimensional ODE of~\cite{bergault2019size} but also a partial derivative term (in the inventory variable) like in models à la Almgren-Chriss. We illustrate our model numerically in Section~\ref{numericSec} using examples from the FX market and discuss the results. In particular, we highlight the existence of a threshold of inventory below which it is not optimal for the dealer to trade in the liquidity pool.\\

\section{The model}
\label{model}

We consider a dealer in charge of a single asset. This asset can be traded by the dealer in two ways: (i)~via requests they receive from clients, through the RFQ or RFS channel, and (ii) via market orders sent to a liquidity pool, for instance an inter-dealer broker platform or an all-to-all platform.\footnote{We can also regard the liquidity pool of our model as an aggregation of numerous liquidity pools.} In the former case, the dealer proposes a price to clients wishing to buy or sell the asset and clients ultimately decide to trade or not to trade. In the latter case, the dealer decides to buy or sell at a market price that depends on the traded volume.\footnote{The model can easily be generalized to multiple assets (see~\cite{gueant2017optimal} for the way to go multi-asset in models \textit{à la} Avellaneda-Stoikov). In that case, each asset may be traded through requests only, market orders only, or both ways.}\\

In Section~\ref{basemodel} we present an optimization problem with state constraints corresponding to risk limits. In Section~\ref{smoothmodel}, we present a slightly modified version that is more practical for mathematical analysis.

\subsection{Modeling framework and notations}
\label{basemodel}
We consider a filtered probability space $\left( \Omega, \mathcal{F},\mathbb{P}; \mathbb{F}= (\mathcal{F}_{t})_{t\geq 0} \right)$ satisfying the usual conditions. We assume this probability space is large enough to support all the processes we introduce.\\

Let us start with the modelling of the price. We model the reference price (for instance the mid-price of the liquidity pool\footnote{For instance, in FX markets, the reference price is typically the mid-price on platforms like EBS or Refinitiv. For corporate bonds, it can be Bloomberg CBBT or MarketAxess CP+.}) of the asset by a process $\left(S_t \right)_{t\geq 0}$. We consider that the dealer can trade continuously in the liquidity pool and we denote by $\left(w_t \right)_{t\geq 0}$ the execution rate of the dealer (they buy when $w_t \geq 0$ and sell otherwise).\footnote{This is similar to what was done in Almgren~\cite{almgren2003optimal}.} Taking into account the permanent market impact of trades within the liquidity pool, the price process $(S_t)_{t\geq 0}$ is modeled as follows: $$dS_t = \sigma dW_t + kw_t dt,$$ with $S_0$ given, where $k$ and $\sigma$ are positive constant and the process $\left(W_t\right)_{t\geq 0}$ is a standard Brownian motion adapted to the filtration $\mathbb{F}$.\footnote{The linearity of the permanent market impact can be justified on theoretical grounds -- see the famous paper~\cite{gatheral2010no} by Gatheral. The Brownian assumption is made in almost all market making papers. It makes perfect sense for the examples we develop in Section~\ref{numericSec} because, given liquidity levels, the relevant time horizon for FX market making is far less than one day (see Figures~\ref{conv_exec} and~\ref{conv_deltas}). For other asset classes, like corporate bonds, the Brownian assumption may be more questionable -- see~\cite{drissi} for the use of Ornstein-Uhlenbeck prices in market making models.}\\

Regarding requests, the dealer proposes bid and ask quotes depending on the size $z\in \mathbb{R}_+^*$ of each request. These quotes are modeled by maps $S^{b},S^{a} : \Omega \times [0,T] \times \mathbb{R}_{+}^{*}\rightarrow \mathbb{R_{+}}$ which are $\mathcal{P} \otimes \mathcal{B}(\mathbb{R}_{+}^{*})$-measurable, where $\mathcal{P}$ denotes the $\sigma$-algebra of $\mathbb{F}$-predictable subsets of $\Omega \times[0,T]$ and $\mathcal{B}(\mathbb{R}_{+}^{*})$ denotes the Borelian sets of $\mathbb{R}_{+}^{*}$.\\

We introduce $J^{b}(dt,dz)$ and $J^{
a}(dt,dz)$ two \textit{càdlàg} $\mathbb{R}_{+}$-marked point processes corresponding to the number of transactions resulting from requests at the bid and at the ask, respectively. In the following sections, we shall denote by $\tilde{J}^{b}(dt,dz)$ and $\tilde{J}^{a}(dt,dz)$ the associated compensated processes.\\

The inventory of the dealer, modeled by the process $(q_{t})_{t\geq 0} $, has the following dynamics:
\begin{equation*}
    dq_t = \int_{\mathbb{R}_{+}^{*}} z J^{b}(dt,dz) - \int_{\mathbb{R}_{+}^{*}} z J^{a}(dt,dz) + w_t  dt,
\end{equation*}
with $q_0$ given. We assume that the processes $J^{b}(dt,dz)$ and $J^{a}(dt,dz)$ do not have simultaneous jumps almost surely. Moreover, we denote by $\left(\nu^{b}_{t}(dz)\right)_{t\geq 0}$ and $\left(\nu^{a}_{t}(dz)\right)_{t\geq 0}$ the intensity kernels of $J^{b}(dt,dz)$ and $J^{
a}(dt,dz)$, respectively.\\ 

We assume that the dealer wants their inventory to remain in an interval $\mathcal Q = [-\tilde q, \tilde q]$, with $\tilde q >0$. The intensities $\left(\nu^{b}_{t}(dz)\right)_{t\geq 0}$ and $\left(\nu^{a}_{t}(dz)\right)_{t\geq 0}$ verify
$$
    \nu^{b}_{t}(dz)\ =\  \mathds{1}_{\left\{q_{t-} + z \in \mathcal Q \right\}} \Lambda^{b}(\delta^{b}(t,z))\mu^{b}(dz) \quad \text{and} \quad
    \nu^{a}_{t}(dz)\ =\  \mathds{1}_{\left\{q_{t-} - z \in \mathcal Q \right\}} \Lambda^{a}(\delta^{a}(t,z))\mu^{a}(dz),$$
where (i) $\mu^{b}$ and $\mu^{a}$ are probability measures on $\mathbb{R}_{+}^{*}$, absolutely continuous with respect to the Lebesgue measure and such that $\int_{\mathbb{R}_{+}^{*}} z\mu^{b}(dz) =: \Delta^{b} < +\infty$ and $\int_{\mathbb{R}_{+}^{*}} z\mu^{a}(dz) =: \Delta^{a} < +\infty$, (ii) $\delta^{b}(t,z)$ and $\delta^{a}(t,z)$ are defined by
\begin{equation}
\delta^{b}(t,z) = S_{t} - S^{b}(t,z) \text{ and }
\delta^{a}(t,z) = S^{a}(t,z) -  S_{t},  \nonumber
\end{equation}
and (iii) $\Lambda^{b}$, $\Lambda^{a}$ are two functions satisfying the following hypotheses $(H_\Lambda)$:
\begin{itemize}
    \item $\Lambda^{b}$ and $\Lambda^{a}$ are twice continuously differentiable,
    \item $\Lambda^{b}$ and $\Lambda^{a}$ are decreasing, with $\forall \delta \in \mathbb{R}$, ${\Lambda^{b}}'(\delta)<0$ and ${\Lambda^{a}}'(\delta)<0$,
    \item $\underset{\delta \rightarrow +\infty}{\lim}\Lambda^{b}(\delta)=0$ and $\underset{\delta \rightarrow +\infty}{\lim} \Lambda^{a}(\delta)=0$,
    \item $\underset{\delta \in \mathbb{R}}{\sup}  \frac{\Lambda^{b}(\delta){\Lambda^{b}}''(\delta)}{\left( {\Lambda^{b}}'(\delta) \right)^{2}}  < 2$ and $\underset{\delta \in \mathbb{R}}{\sup}  \frac{\Lambda^{a}(\delta){\Lambda^{a}}''(\delta)}{\left( {\Lambda^{a}}'(\delta) \right)^{2}}  < 2$. \newline
\end{itemize}

The process $(X_{t})_{t\geq 0}$ modeling the dealer's cash account has the dynamics:
\begin{equation*}
\begin{split}
dX_t & = \int_{\mathbb{R}_+^*}z\Big(S_t + \delta^{a}(t,z) \Big)J^{a}(dt,dz) - \int_{\mathbb{R}_+^*}z\Big(S_t - \delta^{b}(t,z) \Big)J^{b}(dt,dz)-w_t S_t dt - L(w_t)dt \\
& =  \int_{\mathbb{R}_+^*}z\delta^{b}(t,z) J^{b}(dt,dz) + \int_{\mathbb{R}_+^*}z\delta^{a}(t,z) J^{a}(dt,dz) -  L(w_t)dt - S_tdq_t, 
\end{split}    
\end{equation*}
where the penalty function $L:\mathbb{R}\rightarrow \mathbb{R}_+$ (that results from the temporary price impact when trading in the liquidity pool) satisfies the following hypotheses $(H_L)$:
\begin{itemize}
    \item $L(0)=0$,
    \item $L$ is strictly convex, increasing on $\mathbb{R}_+$ and decreasing on $\mathbb{R}_-$,
    \item $L$ is asymptotically superlinear, i.e., $\underset{|\varrho| \rightarrow +\infty}{\lim} \frac{L(\varrho)}{|\varrho|} = +\infty$.
\end{itemize}

For the dealer's inventory to remain in the interval $\mathcal Q = [-\tilde q, \tilde q]$, we define the set of admissible strategies as follows:
\begin{equation}
\mathcal{A}_{\mathcal T} = \Bigg\lbrace w : \Omega \times [0,T] \mapsto \mathbb{R} \bigg| w \text{ is } \mathcal{P}\text{-measurable,} |w_t| \le v_{\infty}\ \mathbb P \otimes dt \  a.e. \text { and } q_t \in \mathcal Q \ \mathbb P \otimes dt \  a.e. \Bigg\rbrace \nonumber
\end{equation}
for a given $v_{\infty}>0$, and
\begin{equation}
\begin{split}
\mathcal{A}_{\mathcal M} = \Bigg\lbrace \delta = \left(\delta^{b},\delta^{a}\right) : \Omega \times [0,T] \times \mathbb{R}_{+}^{*} \mapsto \mathbb{R}^{2} \bigg| \delta \text{ is } \mathcal{P} \otimes \mathcal{B}(\mathbb{R}_{+}^{*})\text{-measurable }\\
\text{and } \delta^{b}(t,z) \wedge \delta^{a}(t,z) \geq -\delta_{\infty}\ \mathbb P \otimes dt \otimes dz \  a.e.  \Bigg\rbrace \nonumber
\end{split}
\end{equation} where $\delta_{\infty}>0$ is a prespecified constant.\\

For two given continuous penalty functions $\psi:\mathbb{R} \rightarrow \mathbb{R}_{+}$ and $\ell:\mathbb{R} \rightarrow \mathbb{R}_{+}$, modeling the risk aversion of the dealer ($\psi$ incentivizes the dealer to avoid large inventories, long or short) and the cost of liquidity, we aim at maximizing

\begin{equation*}
\mathbb{E} \left[ X_{T} +  q_T S_T  - \ell(q_{T}) - \int_{0}^{T} \psi(q_{t}) dt \right]
\end{equation*}

over the set $\mathcal{A}_{\mathcal M} \times \mathcal{A}_{\mathcal T}$ of admissible controls $\left(\delta, w \right)$.

\begin{rem}
For instance, for a constant $\gamma>0$, we can choose $\psi(q)=\frac{\gamma}{2} \sigma^2 q^2$ or $\psi(q) = \gamma \sigma |q|$. For the terminal penalty, it should at least account for the cost associated with complete liquidation of the remaining inventory at time $T$, i.e. $\ell(q) \ge \frac k2 q^2$ (see for instance~\cite{gueant2016financial}).\end{rem}

\subsection{Risk limits and trading costs}
\label{smoothmodel}

The presence of both point processes and absolutely continuous processes in the model requires the use of the theory of discontinuous viscosity solutions to address the control problem with partial differential equation tools. However, the above problem belongs to the family of state-constrained control problems for which the use of (a priori) discontinuous viscosity solutions makes the mathematical analysis unduly complicated, in particular when it comes to proving a comparison principle for the associated Hamilton-Jacobi equation. For that reason, we replace the state constraint by a penalty paid when buying (resp. selling) while already very long (resp. short) and close to the risk limit.\\

More precisely, we fix $\tilde \epsilon \in (0,\tilde q)$ and introduce a function $\zeta: \mathbb R \rightarrow [0,1]$ which is a Lipschitz approximation (we write $L_\zeta$ the Lipschitz constant of $\zeta$) of the indicator function of $\mathbb R_+$, with $$\zeta(q) = 1 \quad \forall q \ge \tilde\epsilon \quad \text{and} \quad \zeta(q) = 0 \quad \forall q \le 0.$$
Then, we use a new control process $$(v_t)_{t \geq 0} \in \tilde{\mathcal{A}}_\mathcal{T} = \Bigg\lbrace v : \Omega \times [0,T] \mapsto \mathbb{R} \bigg| v \text{ is } \mathcal{P}\text{-measurable, and } |v_t| \le v_{\infty}\ \mathbb P \otimes dt \  a.e.\Bigg\rbrace$$ and choose an execution rate of the form  $$w_t = -(v_t)_- \zeta(\tilde{q} + q_{t-}) + (v_t)_+ \zeta(\tilde{q}-q_{t-}).$$
We also introduce a new cost function\footnote{This new function corresponds to the approximation of the cost function  $(w,q) \mapsto L(w) + \infty 1_{q \not\in \mathcal Q}$ by the function $ \tilde{\mathcal L}_0 : \mathbb R^2 \rightarrow \mathbb R \cup \{+\infty\}$ defined by
$$\tilde{\mathcal L}_0 (w,q) = L\left( \frac{-w_-}{\zeta(\tilde{q}+q)} \right) \zeta(\tilde{q}+q) +  L\left( \frac{w_+}{\zeta(\tilde{q}-q)} \right) \zeta(\tilde{q}-q) \quad \forall (w,q) \in \mathbb R^2,$$ with the conventions
$$
L\left( \frac{-w_-}{\zeta(\tilde{q}+q)} \right) \zeta(\tilde{q}+q) =
\begin{cases}
+\infty \quad \text{if } \zeta(\tilde{q}+q) = 0 \text{ and } w_- \neq 0\\
0 \quad \text{if } \zeta(\tilde{q}+q) = 0 \text{ and } w_- = 0,
\end{cases}
$$
and similarly
$$
L\left( \frac{w_+}{\zeta(\tilde{q}-q)} \right) \zeta(\tilde{q}-q) =
\begin{cases}
+\infty \quad \text{if } \zeta(\tilde{q}-q) = 0 \text{ and } w_+ \neq 0\\
0 \quad \text{if } \zeta(\tilde{q}-q) = 0 \text{ and } w_+ = 0.
\end{cases}
$$
These conventions are natural as $L$ is assumed to be asymptotically superlinear.} $\tilde{\mathcal L}: \mathbb R^2 \rightarrow \mathbb R$ given by
$$\tilde{\mathcal L} (v,q) = L\left( -v_- \right) \zeta(\tilde q + q) +  L\left(v_+ \right) \zeta(\tilde q -q) \quad \forall (v,q) \in \mathbb R^2 .$$

It is important to notice that the control process and the associated costs are exactly the same as the initial one whenever the inventory lies in $[-\tilde q + \tilde\epsilon, \tilde q - \tilde\epsilon]$. As the value of the inventory gets closer and closer to $-\tilde q$ (in $(-\tilde q , -\tilde q + \tilde\epsilon)$), the dealer can only sell smaller and smaller volumes and pays an increasingly high cost to sell a given volume, and, symmetrically, as the value of the inventory gets closer and closer to $\tilde q$ (in $(\tilde q - \tilde\epsilon, \tilde q)$), the dealer can only buy smaller and smaller volumes and pays an increasingly high cost to buy a given volume. As a result, the inventory naturally stays in $\mathcal Q$. Because the penalty function $\psi$ incentivizes the dealer to avoid large inventories in absolute value, there is a priori no reason for the dealer to buy (resp. sell) in liquidity pools when they are very long (resp. short). Therefore, when risk limits are high enough, our slight modification of the model has absolutely no impact on the optimal strategy.\\

The price process $(S_t)_{t\geq 0}$ has subsequently the following dynamics: $$dS_t = \sigma dW_t + k\left(-(v_t)_- \zeta(\tilde{q} + q_{t-}) + (v_t)_+ \zeta(\tilde{q}-q_{t-})\right) dt,$$
and the inventory $(q_t)_{t \ge 0}$ has the dynamics:
\begin{equation*}
    dq_t = \int_{\mathbb{R}_{+}^{*}} z J^{b}(dt,dz) - \int_{\mathbb{R}_{+}^{*}} z J^{a}(dt,dz) + \left(-(v_t)_- \zeta(\tilde{q} + q_{t-}) + (v_t)_+ \zeta(\tilde{q}-q_{t-}) \right)  dt.
\end{equation*}
Finally, the process $(X_{t})_{t\geq 0}$ modeling the dealer's cash account has the dynamics:
\begin{equation*}
\begin{split}
dX_t & = \int_{\mathbb{R}_+^*}z\Big(S_t + \delta^{a}(t,z) \Big)J^{a}(dt,dz) - \int_{\mathbb{R}_+^*}z\Big(S_t - \delta^{b}(t,z) \Big)J^{b}(dt,dz) \\
&\quad - \left(-(v_t)_- \zeta(\tilde{q} + q_{t-}) + (v_t)_+ \zeta(\tilde{q}-q_{t-})\right) S_t dt - \tilde{\mathcal L}\big(v_t, q_{t-} \big)dt \\
& = \int_{\mathbb{R}_+^*}z\delta^{b}(t,z) J^{b}(dt,dz) + \int_{\mathbb{R}_+^*}z\delta^{a}(t,z) J^{a}(dt,dz)  \\
&\quad -  \tilde{\mathcal L}\big(v_t, q_{t-} \big)dt - S_tdq_t. 
\end{split}    
\end{equation*}

The resulting optimization problem is that of maximizing
\begin{equation*}
\mathbb{E} \left[ X_{T} +  q_T S_T  - \ell(q_{T}) - \int_{0}^{T} \psi(q_{t}) dt \right] = \mathbb{E} \left[ X_{T-} +  q_{T-} S_T  - \ell(q_{T-}) - \int_{0}^{T} \psi(q_{t-}) dt \right]
\end{equation*}
over the modified set $\mathcal A:= \mathcal A_{\mathcal M} \times \tilde{\mathcal A}_{\mathcal T}$ of admissible controls $\left(\delta, v \right)$.\\

After applying Itô's formula to $\left(X_{t} + q_t S_t  \right)_{t\geq 0}$ between $0$ and $T-$, it is easy to see that the problem is equivalent to maximizing:
\begin{equation}
\begin{split}
\mathbb{E}&\left[\int\limits_{0}^{T} \Bigg\lbrace \int_{\mathbb{R}_{+}^{*}} \Big(z\delta^{b}(t,z) \mathds{1}_{\left\{q_{t-} + z \in \mathcal Q \right\}} \Lambda^{b}(\delta^{b}(t,z))\mu^{b}(dz) + z \delta^{a}(t,z) \mathds{1}_{\left\{q_{t-} - z \in \mathcal Q \right\}} \Lambda^{a}(\delta^{a}(t,z))\mu^{a}(dz) \Big) \right. \nonumber \\
&\left. \qquad + k q_t \left(-(v_t)_- \zeta(\tilde{q} + q_{t-}) + (v_t)_+ \zeta(\tilde{q}-q_{t-}) \right) - \tilde{\mathcal L}\big(v_t, q_{t-} \big) - \psi(q_{t-}) \Bigg\rbrace dt -\ell(q_{T-}) \vphantom{\int\limits_{0}^{T}} \right],\nonumber
\end{split}
\end{equation}
over the set of admissible controls $\mathcal{A}$.\\

We therefore introduce the function $\mathcal{J}:[0,T] \times \mathcal Q\times \mathcal{A}_{\mathcal M} \times \tilde{\mathcal A}_{\mathcal T}\longrightarrow \mathbb{R}$ such that, for all $t \in [0,T]$, for all $q \in \mathcal Q$ and for all $\left(\bar \delta, \bar v \right) \in \mathcal{A}$, if we denote by $\big(q^{t,q,\bar{\delta}, \bar{v}}_{s}\big)_{s\geq t}$ the inventory process starting in state $q$ at time $t$ and controlled by $\left(\bar{\delta},\bar{v} \right)$:
\begin{equation}
\begin{split}
\mathcal{J}\left(t,q,\bar{\delta},\bar{v}\right) = \ & \mathbb{E} \left[\int\limits_{t}^{T} \Bigg\lbrace \int_{\mathbb{R}_{+}^{*}} \Big(z\delta^{b}(s,z)  \mathds{1}_{\left\{q^{t,q,\bar{\delta}, \bar{v}}_{s-} + z \in \mathcal Q \right\}} \Lambda^{b}(\delta^{b}(s,z))\mu^{b}(dz) \right. \\
& \qquad + z\delta^{a}(s,z) \mathds{1}_{\left\{q^{t,q,\bar{\delta}, \bar{v}}_{s-} - z \in \mathcal Q \right\}}\Lambda^{a}(\delta^{a}(s,z))\mu^{a}(dz) \Big)\\
&\left. \qquad + k q^{t,q,\bar{\delta}, \bar{v}}_{s-} \left(-(v_s)_- \zeta(\tilde q + q^{t,q,\bar{\delta}, \bar{v}}_{s-}) + (v_s)_+ \zeta(\tilde q - q^{t,q,\bar{\delta}, \bar{v}}_{s-}) \right) \right.\\
& \qquad \left. - \tilde{\mathcal L}\big(v_s, q^{t,q,\bar{\delta}, \bar{v}}_{s-} \big) - \psi(q^{t,q,\bar{\delta}, \bar{v}}_{s-}) \Bigg\rbrace ds -\ell(q^{t,q,\bar{\delta}, \bar{v}}_{T-}) \vphantom{\int\limits_{t}^{T}} \right].\nonumber
\end{split}
\end{equation}

The value function $\theta:[0,T]\times \mathcal Q \rightarrow \mathbb{R}$ of the problem is then defined as follows:

\begin{equation}
\theta(t,q) = \underset{\left(\bar{\delta}, \bar{v} \right)\in \mathcal{A}}{\sup} \mathcal{J}\left(t,q,\bar{\delta}, \bar{v}\right), \quad \forall (t,q) \in [0,T] \times \mathcal Q. \nonumber
\end{equation}

Using the theory of discontinuous viscosity solutions, we will show that $\theta$ is in fact the unique continuous viscosity solution on $[0,T]\times \mathcal Q$ to the following Hamilton-Jacobi partial integro-differential equation:
\begin{equation}
\tag{HJ}
\begin{cases}
0 = -\partial_t \theta(t,q)+ \psi(q) - \text{\scalebox{0.6}[1]{$\bigint$}}_{\mathbb{R}_{+}^{*}} \mathds{1}_{\left\{q + z \in \mathcal Q \right\}} zH^{b} \left(\frac{\theta(t,q) -  \theta(t,q+z) }{z}\right) \mu^{b}(dz)\\
\quad - \text{\scalebox{0.6}[1]{$\bigint$}}_{\mathbb{R}_{+}^{*}} \mathds{1}_{\left\{q- z \in \mathcal Q \right\}}zH^{a} \left(\frac{\theta(t,q) - \theta(t,q-z)}{z} \right) \mu^{a}(dz) -  \mathcal H \left(\partial_{q}\theta(t,q), q \right)\quad  \forall t \in [0,T) \\
\theta(T,q) = -\ell(q),
\end{cases}
\label{eqn:HJB}
\end{equation}
where
\begin{equation}
H^{b}:p\in\mathbb{R} \mapsto \underset{\delta \geq -\delta_{\infty}}{\sup} \Lambda^{b}(\delta)(\delta-p) \text{ and }H^{a}:p\in\mathbb{R} \mapsto  \underset{\delta \geq -\delta_{\infty}}{\sup} \Lambda^{a}(\delta)(\delta-p) ,\nonumber
\end{equation}
and
\begin{equation}
    \mathcal H:(p,q)\in\mathbb{R} \times \mathcal Q \mapsto \underset{|v| \le v_{\infty}}{\sup} \left(-v_- \zeta(\tilde q + q) + v_+ \zeta(\tilde q - q) \right)(p+kq) - \tilde{\mathcal L}\big(v, q \big).\nonumber
\end{equation}

\section{Mathematical analysis}
\label{viscoSec}

\subsection{Preliminary results}

We first recall a result (Lemma~\ref{lemmH}) which is proved in~\cite{bergault2019size}.\\

\begin{lemme}
\label{lemmH}
$H^{b}$ and $H^{a}$ are two continuously differentiable decreasing functions and the supremum in the definition of $H^{b}(p)$ (respectively $H^{a}(p)$) is reached at a unique $\delta^{b*}(p)$ (respectively $\delta^{a*}(p)$). Furthermore, $\delta^{b*}$ and $\delta^{a*}$ are continuous and nondecreasing in $p$.\newline
\end{lemme}

We then state another useful lemma:
\begin{lemme}
\label{zHbounded}
Let $\varphi:[0,T]\times \mathcal{Q} \mapsto \mathbb{R}$ be a bounded function. The functions $$(t,q,z) \in [0,T] \times \mathcal{Q} \times \mathbb{R}_{+}^{*} \mapsto \mathds{1}_{\{q +z \in \mathcal{Q}\}}zH^{b} \left( \frac{\varphi(t,q) - \varphi(t,q+z)}{z} \right)$$ and $$(t,q,z)\in [0,T] \times \mathcal{Q} \times \mathbb{R}_{+}^{*} \mapsto \mathds{1}_{\{q -z \in \mathcal{Q}\}}zH^{a} \left(\frac{\varphi(t,q) - \varphi(t,q-z)}{z} \right)$$ are bounded. \newline
\end{lemme}
\begin{proof}
We only prove it for the ask side (the proof for the bid side is similar).\\

Let $t\in[0,T],$ $q\in \mathcal{Q}$ and $z \in \mathbb{R}_{+}^{*} $ such that $q-z \in \mathcal{Q}$.\\

We have
\begin{equation}
\begin{split}
zH^{a} \left(\frac{\varphi(t,q) - \varphi(t,q-z)}{z} \right) & = z \underset{\delta \geq -\delta_{\infty}}{\sup} \Lambda^{a}(\delta)\bigg(\delta-\frac{\varphi(t,q) - \varphi(t,q-z)}{z}\bigg)  \\
& \leq z \underset{\delta \geq -\delta_{\infty}}{\sup} \Lambda^{a}(\delta)\delta + \underset{\delta \geq -\delta_{\infty}}{\sup}  - \Lambda^{a}(\delta)\bigg(\varphi(t,q) - \varphi(t,q-z) \bigg) \\
& \leq 2\tilde{q} H^a(0)  + 2 \Lambda^{a}(-\delta_{\infty}) \underset{[0,T]\times \mathcal{Q}}{\sup}\ |\varphi|. \nonumber
\end{split}
\end{equation}

For the other bound, we have
\begin{eqnarray*}
zH^{a} \left(\frac{\varphi(t,q) - \varphi(t,q-z)}{z} \right) &=& \underset{\delta \geq -\delta_{\infty}}{\sup} \Bigg\lbrace z\Lambda^{a}(\delta)\delta - \Lambda^{a}(\delta)\bigg(\varphi(t,q) - \varphi(t,q-z) \bigg) \Bigg\rbrace.\\
&\ge & - \Lambda^{a}(0)\bigg(\varphi(t,q) - \varphi(t,q-z) \bigg)\\
&\ge &  - 2\Lambda^{a}(0)\underset{[0,T]\times \mathcal{Q}}{\sup}\ |\varphi|.\\
\end{eqnarray*}
\end{proof}

We can now state a first simple result about the value function $\theta$:
\begin{prop}
The value function $\theta$ is bounded on $ [0,T] \times \mathcal Q.$
\end{prop}

\begin{proof}
$\forall\ \left(t,q,\bar \delta, \bar v \right) \in [0,T] \times \mathcal Q \times\mathcal{A}_{\mathcal M} \times \tilde{\mathcal{A}}_{\mathcal T},$ we have 
\begin{equation}
\begin{split}
\mathcal{J}\left(t,q,\bar{\delta},\bar{v}\right) = & \mathbb{E} \left[\int\limits_{t}^{T} \Bigg\lbrace \int_{\mathbb{R}_{+}^{*}} \Big(z\delta^{b}(s,z) \mathds{1}_{\left\{q^{t,q,\bar{\delta}, \bar{v}}_{s-} + z \in \mathcal Q \right\}} \Lambda^{b}(\delta^{b}(s,z))\mu^{b}(dz)\right.\\
&+ z\delta^{a}(s,z) \mathds{1}_{\left\{q^{t,q,\bar{\delta}, \bar{v}}_{s-} -z \in \mathcal Q \right\}} \Lambda^{a}(\delta^{a}(s,z))\mu^{a}(dz) \Big)\\
&\left.+ k q^{t,q,\bar{\delta}, \bar{v}}_{s-} \left(-(v_s)_- \zeta(\tilde q + q^{t,q,\bar{\delta}, \bar{v}}_{s-}) + (v_s)_+ \zeta(\tilde q - q^{t,q,\bar{\delta}, \bar{v}}_{s-}) \right)\right.\\
&\left. - \tilde{\mathcal L}\big(v_s, q^{t,q,\bar{\delta}, \bar{v}}_{s-} \big) - \psi(q^{t,q,\bar{\delta}, \bar{v}}_{s-}) \Bigg\rbrace ds -\ell(q^{t,q,\bar{\delta}, \bar{v}}_{T-}) \vphantom{\int\limits_{t}^{T}} \right].\nonumber
\end{split}
\end{equation}
As $\ell,$ $\psi$, and $\tilde{\mathcal L}$ are nonnegative, we get
\begin{equation}
\begin{split}
\mathcal{J}\left(t,q,\bar{\delta},\bar{v}\right)& \leq   \mathbb{E} \left[\int\limits_{t}^{T} \Bigg\lbrace \int_{\mathbb{R}_{+}^{*}}  \Big(z\delta^{b}(s,z) \mathds{1}_{\left\{q^{t,q,\bar{\delta}, \bar{v}}_{s-} + z \in \mathcal Q \right\}} \Lambda^{b}(\delta^{b}(s,z))\mu^{b}(dz) \right.\\
&\qquad + z\delta^{a}(s,z) \mathds{1}_{\left\{q_{s-} - z \in \mathcal Q \right\}} \Lambda^{a}(\delta^{a}(s,z))\mu^{a}(dz) \Big)\\
& \qquad + k q^{t,q,\bar{\delta}, \bar{v}}_{s-} \left(-(v_s)_- \zeta(\tilde q + q^{t,q,\bar{\delta}, \bar{v}}_{s-})  \left. \vphantom{\int\limits_{t}^{T} \Bigg\lbrace \int_{\mathbb{R}_{+}^{*}}  \Big(z\delta^{b}(s,z) \mathds{1}_{\left\{q^{t,q,\bar{\delta}, \bar{v}}_{s-} + z \in \mathcal Q \right\}} \Lambda^{b}(\delta^{b}(s,z))\mu^{b}(dz)}+ (v_s)_+ \zeta(\tilde q - q^{t,q,\bar{\delta}, \bar{v}}_{s-}) \right) \Bigg\rbrace ds   \right] \\
& \leq T  \left( \Delta^{b} \underset{\delta \geq -\delta_{\infty}}{\sup} \delta \Lambda^{b}(\delta) + \Delta^{a} \underset{\delta \geq -\delta_{\infty}}{\sup} \delta \Lambda^{a}(\delta) + k v_{\infty } \tilde q\right)\nonumber\\
& \leq T  \left( \Delta^{b} H^b(0)+ \Delta^{a} H^a(0)+ k v_{\infty } \tilde q\right).\nonumber 
\end{split}
\end{equation}
The right-hand side is independent from $t,$ $q,$ $\bar{\delta}$ and $\bar{v},$ so it is clear that $\mathcal{J}$ and $\theta$ are bounded from above.\\

By considering the control process corresponding to $\delta^b(t,z) = \delta^a(t,z) = v_t = 0$, we obtain  
\begin{equation*}
\begin{split}
\theta(t,q)  &\geq - T \psi(q) - \ell(q).
\end{split}    
\end{equation*}
As $\psi$ and $\ell$ are continuous and $\mathcal Q$ is compact, we get that $\theta$ is bounded from below.
\end{proof}

Turning to the Hamiltonian function associated with the possibility to trade in a liquidity pool, we will need the following lemma:
\begin{lemme}
\label{lemmHronde}
$\mathcal H$ is a continuous function that satisfies
$$\exists C_{\mathcal H}>0, \forall p \in \mathbb R,  \forall q,y \in \mathcal Q, \left| \mathcal H (p,q) - \mathcal H (p,y) \right| \le C_{\mathcal H} (1+|p|) |q-y|.$$
Furthermore, the supremum in the definition of $\mathcal H(p,q)$ is reached for $v^{*}(p,q) =\bar{\mathcal H}'(p+kq)$, where $$\bar{\mathcal H} (r) =  \underset{|v| \le v_{\infty}}{\sup} \left(r v - L(v) \right).$$
\end{lemme}

\begin{proof}

Let us first recall from classical results of convex analysis that, given the hypotheses $(H_L)$, $\bar{\mathcal H}$ is a continuously differentiable function that satisfies $\bar{\mathcal H}(0) = 0$. Moreover, the supremum in the definition of $\bar{\mathcal H}(r)$ is reached uniquely at $\bar{\mathcal H}'(r)$. In particular, $\bar{\mathcal H}$ is a Lipschitz function with Lipschitz constant equal to~$v_\infty$.\\  

For all $p \in \mathbb R$ and $q \in \mathcal Q$, we have $\mathcal H (p,q) = \underset{|v| \le v_{\infty}}{\sup} \left(\left(-v_- \zeta(\tilde q + q) + v_+ \zeta(\tilde q - q) \right)(p+kq) - \tilde{\mathcal L}(v, q) \right)$.\\

Let us consider first that $p+kq \ge 0$. If $v < 0$, then $\left(-v_- \zeta(\tilde q + q) + v_+ \zeta(\tilde q - q) \right)(p+kq) - \tilde{\mathcal L}(v, q) \le 0$. Therefore, when $p+kq \ge 0$, we have $$\mathcal H (p,q) = \zeta(\tilde q - q) \underset{0 \le v \le v_{\infty}}{\sup} \left((p+kq) v - L(v) \right) = \zeta(\tilde q - q) \underset{|v| \le v_{\infty}}{\sup} \left((p+kq) v - L(v) \right) = \zeta(\tilde q - q) \bar{\mathcal H}(p+kq)$$ and the supremum is reached for $v^*(p,q) = \bar{\mathcal H}'(p+kq)$.\\

Let us now come to the case $p+kq \le 0$. If $v > 0$, then $\left(-v_- \zeta(\tilde q + q) + v_+ \zeta(\tilde q - q) \right)(p+kq) - \tilde{\mathcal L}(v, q) \le 0$. Therefore, when $p+kq \le 0$, we have $$\mathcal H (p,q) = \zeta(\tilde q + q) \underset{-v_{\infty} \le v \le 0}{\sup} \left((p+kq) v - L(v) \right) = \zeta(\tilde q + q) \underset{|v| \le v_{\infty}}{\sup} \left((p+kq) v - L(v) \right) = \zeta(\tilde q + q) \bar{\mathcal H}(p+kq)$$ and the supremum is reached for $v^*(p,q) = \bar{\mathcal H}'(p+kq)$.\\

Overall, because $\bar{\mathcal H}(0) = 0$, we can write $\mathcal H (p,q) = \zeta(\tilde q - q)\bar{\mathcal H} \left((p+kq)_+\right) + \zeta(\tilde q + q) \bar{\mathcal H} \left(-(p+kq)_- \right)$ and consequently $\mathcal H$ is continuous.\\

Now, let us take $p \in \mathbb R$ and $q,y \in \mathcal Q$. We have 

\begin{eqnarray*}
    |\mathcal H(p,q) - \mathcal H(p,y)| & \le &\left|\zeta(\tilde q - q)\bar{\mathcal H} \left((p+kq)_+\right) - \zeta(\tilde q - y)\bar{\mathcal H} \left((p+ky)_+\right) \right|\\ &&+ \left|\zeta(\tilde q + q)\bar{\mathcal H} \left(-(p+kq)_-\right) - \zeta(\tilde q + y)\bar{\mathcal H} \left(-(p+ky)_-\right) \right| \\
    & \le &  \left|(\zeta(\tilde q - q)-\zeta(\tilde q - y))\bar{\mathcal H} \left((p+kq)_+\right) - \zeta(\tilde q - y)\left(\bar{\mathcal H} \left((p+ky)_+\right) - \bar{\mathcal H} \left((p+kq)_+\right)\right) \right|\\
    && + \left|(\zeta(\tilde q + q)-\zeta(\tilde q + y))\bar{\mathcal H} \left(-(p+kq)_-\right) - \zeta(\tilde q + y)\left(\bar{\mathcal H} \left(-(p+ky)_-\right) - \bar{\mathcal H} \left(-(p+kq)_-\right)\right) \right|\\
    & \le & |\zeta(\tilde q - q)-\zeta(\tilde q - y)|\bar{\mathcal H} \left((p+kq)_+\right) + \left|\bar{\mathcal H} \left((p+ky)_+\right) - \bar{\mathcal H} \left((p+kq)_+\right)\right|\\
    && + |\zeta(\tilde q + q)-\zeta(\tilde q + y)|\bar{\mathcal H} \left(-(p+kq)_-\right) + \left|\bar{\mathcal H} \left(-(p+ky)_-\right) - \bar{\mathcal H} \left(-(p+kq)_-\right)\right|\\
    & \le & L_{\zeta} |q-y| \left( \bar{\mathcal H} \left((p+kq)_+\right) + \bar{\mathcal H} \left(-(p+kq)_-\right)\right)\\
    && + \left|\bar{\mathcal H} \left((p+ky)_+\right) - \bar{\mathcal H} \left((p+kq)_+\right)\right| + \left|\bar{\mathcal H} \left(-(p+ky)_-\right) - \bar{\mathcal H} \left(-(p+kq)_-\right)\right|\\
\end{eqnarray*}
As $\bar{\mathcal H}$ is Lipschitz with Lipschitz constant $v_\infty$ and $\bar{\mathcal H}(0)=0$, we get
\begin{align*}
    |\mathcal H(p,q) - \mathcal H(p,y)| &  \le L_{\zeta} |q-y| v_{\infty}(|p|+ k|q|) + 2 v_\infty k |q-y|\\
    & \le C_{\mathcal H} (1+|p|)|q-y|,
\end{align*}
for $C_{\mathcal H} = \max\left(L_{\zeta}v_\infty, L_{\zeta}v_\infty k \tilde{q} + 2 v_\infty k\right)$.\\
\end{proof}

\subsection{Existence result}

We denote by $C^1 := C^1 \left([0,T) \times \mathbb R \right)$ the class of functions $\varphi:[0,T) \times \mathbb R  \rightarrow \mathbb{R}$ that are continuously differentiable on $[0,T) \times \mathbb R  .$ 

\begin{defi}
(i) If $u$ is an upper semicontinuous (USC) function on $[0,T]\times \mathcal{Q}$, we say that $u$ is a viscosity subsolution to \eqref{eqn:HJB} on $[0,T) \times \mathcal{Q}$ if $\forall (\bar{t},\bar{q}) \in [0,T) \times \mathcal{Q}$, $\forall \varphi \in C^1$ such that $(u-\varphi)(\bar{t},\bar{q}) = \underset{(t,q) \in [0,T) \times \mathcal{Q}}{\max}(u-\varphi)(t,q)$, we have:
\begin{equation}
\begin{split}
-\frac{\partial \varphi}{\partial t}(\bar{t}, \bar{q}) &+ \psi(\bar{q}) - \int_{\mathbb{R}_{+}^{*}}  \mathds{1}_{\{ \bar{q} + z \in \mathcal Q\}} zH^{b} \left(\frac{\varphi(\bar{t}, \bar{q}) -  \varphi(\bar{t}, \bar{q}+z) }{z}\right) \mu^{b}(dz)  \\
&-  \int_{\mathbb{R}_{+}^{*}}  \mathds{1}_{\{ \bar{q} - z \in \mathcal Q\}} zH^{a} \left(\frac{\varphi(\bar{t}, \bar{q}) - \varphi(\bar{t}, \bar{q}-z)}{z} \right) \mu^{a}(dz) - \mathcal H \left(\partial_{q} \varphi(\bar{t},\bar{q}), \bar{q} \right) \ \leq 0. \nonumber
\end{split}
\end{equation}

(ii) If $v$ is a lower semicontinuous (LSC) function on $[0,T]\times \mathcal{Q}$, we say that $v$ is a viscosity supersolution to \eqref{eqn:HJB} on $[0,T) \times \mathcal{Q}$ if $\forall (\bar{t},\bar{q}) \in [0,T) \times \mathcal{Q}$, $\forall \varphi \in C^1$ such that $(v-\varphi)(\bar{t},\bar{q}) = \underset{(t,q) \in [0,T) \times \mathcal{Q}}{\min}(v-\varphi)(t,q)$, we have:
\begin{equation}
\begin{split}
-\frac{\partial \varphi}{\partial t}(\bar{t}, \bar{q}) &+ \psi(\bar{q}) - \int_{\mathbb{R}_{+}^{*}}  \mathds{1}_{\{ \bar{q} + z \in \mathcal Q\}} zH^{b} \left(\frac{\varphi(\bar{t}, \bar{q}) -  \varphi(\bar{t}, \bar{q}+z) }{z}\right) \mu^{b}(dz)  \\
&- \int_{\mathbb{R}_{+}^{*}} \mathds{1}_{\{ \bar{q} - z \in \mathcal Q\}}  zH^{a} \left(\frac{\varphi(\bar{t}, \bar{q}) - \varphi(\bar{t}, \bar{q}-z)}{z} \right) \mu^{a}(dz) -  \mathcal H \left(\partial_{q} \varphi(\bar{t},\bar{q}),\bar{q} \right) \ \geq 0. \nonumber
\end{split}
\end{equation}

(iii) If $\theta$ is locally bounded on $[0,T) \times \mathcal{Q}$, we say that $\theta$ is a viscosity solution to \eqref{eqn:HJB} on $[0,T) \times \mathcal{Q}$ if its upper semicontinuous envelope $\theta^{*}$ and its lower semicontinuous envelope $\theta_{*}$ are respectively subsolution on $[0,T) \times \mathcal Q$ and supersolution on $[0,T) \times  \mathcal Q$ to \eqref{eqn:HJB}.\newline
\end{defi}

The following result is proved in the appendix: \\
\begin{prop}
\label{eqvisco}
(i) Let $u$ be a $USC$ function on $[0,T]\times \mathcal Q$. $u$ is a viscosity subsolution to \eqref{eqn:HJB} on $[0,T) \times \mathcal{Q}$ if and only if $\forall (\bar{t},\bar{q})\in [0,T) \times \mathcal Q$, $\forall \varphi \in \mathcal{C}^1$ such that $\underset{(t,q) \in [0,T) \times \mathcal Q}{\max}(u-\varphi)(t,q) = (u-\varphi)(\bar{t},\bar{q})$, we have
\begin{equation}
\begin{split}
-\frac{\partial \varphi}{\partial t}(\bar{t}, \bar{q}) &+ \psi(\bar{q}) -  \int_{\mathbb{R}_{+}^{*}} \mathds{1}_{\{ \bar{q} + z \in \mathcal Q\}} zH^{b} \left(\frac{u(\bar{t}, \bar{q}) -  u(\bar{t}, \bar{q}+z) }{z}\right) \mu^{b}(dz)  \\
&-  \int_{\mathbb{R}_{+}^{*}} \mathds{1}_{\{ \bar{q} - z \in \mathcal Q\}} zH^{a} \left(\frac{u(\bar{t}, \bar{q}) - u(\bar{t}, \bar{q}-z)}{z} \right) \mu^{a}(dz) -  \mathcal H \left(\partial_{q} \varphi(\bar{t},\bar{q}) ,\bar{q} \right) \ \leq 0. \nonumber
\end{split}
\end{equation}

(ii) Let $v$ be a $LSC$ function on $[0,T]\times \mathcal Q$. $v$ is a viscosity supersolution to \eqref{eqn:HJB} on $[0,T) \times \mathcal{Q}$ if and only if $\forall (\bar{t},\bar{q})\in [0,T) \times \mathcal Q$, $\forall \varphi \in \mathcal{C}^1$ such that $\underset{(t,q) \in [0,T) \times \mathcal Q}{\min}(v-\varphi)(t,q) = (v-\varphi)(\bar{t},\bar{q})$, we have:

\begin{equation}
\begin{split}
-\frac{\partial \varphi}{\partial t}(\bar{t}, \bar{q}) &+ \psi(\bar{q}) -  \int_{\mathbb{R}_{+}^{*}}  \mathds{1}_{\{ \bar{q} + z \in \mathcal Q\}} zH^{b} \left(\frac{v(\bar{t}, \bar{q}) -  v(\bar{t}, \bar{q}+z) }{z}\right) \mu^{b}(dz)  \\
&-  \int_{\mathbb{R}_{+}^{*}} \mathds{1}_{\{ \bar{q} - z \in \mathcal Q\}} zH^{a} \left(\frac{v(\bar{t}, \bar{q}) - v(\bar{t}, \bar{q}-z)}{z} \right) \mu^{a}(dz) -  \mathcal H \left(\partial_{q} \varphi(\bar{t},\bar{q}),\bar{q} \right) \ \geq 0. \nonumber
\end{split}
\end{equation}
\newline
\end{prop}

We can now prove that $\theta$ is a viscosity solution to \eqref{eqn:HJB}.

\begin{prop}
$\theta$ is a viscosity subsolution to \eqref{eqn:HJB} on $[0,T) \times \mathcal Q$.\newline
\end{prop}

\begin{proof}
$\theta$ is bounded on $[0,T]\times \mathcal Q $ so we can define $\theta^{*}$ its upper semicontinuous envelope.\\

Let $(\bar{t},\bar{q}) \in [0,T) \times \mathcal Q$ and $\varphi \in C^1$ such that
\begin{equation}
0 = (\theta^{*} - \varphi)(\bar{t},\bar{q}) = \underset{(t,q)\in[0,T) \times \mathcal Q}{\max}\ (\theta^{*} - \varphi)(t,q). \nonumber
\end{equation}

We can classically assume this maximum to be strict. By definition of $\theta^{*}(\bar{t},\bar{q})$, their exists $(t_{m},q_{m})_{m}$ a sequence of $[0,T) \times \mathcal Q$ such that
\begin{equation}
\begin{split}
(t_{m},q_{m}) \xrightarrow[{m \rightarrow +\infty}]{} (\bar{t},\bar{q}),\\
\theta(t_{m},q_{m}) \xrightarrow[{m \rightarrow +\infty}]{} \theta^{*}(\bar{t},\bar{q}). \nonumber
\end{split}
\end{equation}

We prove the result by contradiction. Assume there exists $\eta > 0$ such that
\begin{equation}
\begin{split}
-\frac{\partial \varphi}{\partial t}(\bar{t}, \bar{q}) &+ \psi(\bar{q}) - \int_{\mathbb{R}_{+}^{*}} \mathds{1}_{\{ \bar{q} + z \in \mathcal Q\}} zH^{b} \left(\frac{\varphi(\bar{t}, \bar{q}) -  \varphi(\bar{t}, \bar{q}+z) }{z}\right) \mu^{b}(dz)  \\
&- \int_{\mathbb{R}_{+}^{*}}  \mathds{1}_{\{ \bar{q} - z \in \mathcal Q\}} zH^{a} \left( \frac{\varphi(\bar{t}, \bar{q}) - \varphi(\bar{t}, \bar{q}-z)}{z} \right) \mu^{a}(dz) -   \mathcal H \left(\partial_{q} \varphi(\bar{t},\bar q),\bar{q} \right) \ > 2\eta. \nonumber
\end{split}
\end{equation}

Then, as $\varphi$ is continuously differentiable and $\mu^{b}$ and $\mu^{a}$ absolutely continuous with respect to the Lebesgue measure, we must have
\begin{equation}
\label{eq2}
\begin{split}
-\frac{\partial \varphi}{\partial t}(t,q) &+ \psi(q) -  \int_{\mathbb{R}_{+}^{*}}  \mathds{1}_{\{ q + z \in \mathcal Q\}} zH^{b} \left(\frac{\varphi(t,q) -  \varphi(t,q+z) }{z}\right) \mu^{b}(dz)  \\
&- \int_{\mathbb{R}_{+}^{*}}  \mathds{1}_{\{ q - z \in \mathcal Q\}} zH^{a} \left(\frac{\varphi(t,q) - \varphi(t,q-z)}{z} \right) \mu^{a}(dz)  -   \mathcal H \left(\partial_{q} \varphi(t,q),q \right) \ \geq 0
\end{split}
\end{equation}
on $B:=\big((\bar{t}-r,\bar{t}+r)\cap [0,T) \big) \times \left( (\bar{q}-r,\bar q + r) \cap \mathcal Q \right)$ for a sufficiently small $r \in \big(0,T-\bar{t} \big)$. Without loss of generality, we can assume that $B$ contains the sequence $(t_{m},q_{m})_{m}$.\newline

Then, by potentially reducing the value of $\eta$, we have
\begin{equation}
\theta \leq \theta^{*} \leq \varphi - \eta \nonumber
\end{equation}
on the parabolic boundary $\partial_{p}B$ of $B$, i.e. 
$$\partial_{p}B = \Big( \big((\bar{t}-r,\bar{t}+r)\cap [0,T)\big) \times \left( \{\bar{q}-r,\bar q + r\} \cap \mathcal Q \right) \Big) \cup \Big( \{\bar{t}+r\}~\times~\overline{(\bar{q}-r,\bar q + r) \cap \mathcal Q} \Big).$$ Without loss of generality we can assume that the above inequality holds on
\begin{equation}
\tilde{B}:= \{ (t,q+z)\ |\ (t,q,z) \in B \times \mathbb R,\  q+z \in (\bar{q}-r,\bar q + r)^{c} \cap \mathcal Q \}, \nonumber
\end{equation}
which is also bounded.\newline

We introduce an arbitrary control $ \delta=(\delta^{b},\delta^{a})\in \mathcal{A}_{\mathcal M}$. We also introduce an arbitrary control $ v \in \mathcal{A}_{\mathcal T}$. We denote by $\pi_{m}$ the first exit time of $(t,q^{m}_{t})_{t\geq t_{m}}$ from $B$ (where $q^{m}_{t}:=q^{t_{m},q_{m}, {\delta}, {v}}_{t}$): $
\pi_{m} = \inf \{t \geq t_{m} | \big(t, q^{m}_{t} \big) \not\in B \}.$\\

By Itô's formula,
\begin{equation}
\begin{split}
\varphi(\pi_{m}, q^{m}_{\pi_{m}}) & = \varphi(t_{m},q_{m}) + \int\limits_{t_{m}}^{\pi_{m}} \frac{\partial \varphi}{\partial t}(s,q^{m}_{s-}) ds\\
& + \int\limits_{t_{m}}^{\pi_{m}} \int_{\mathbb{R}_{+}^{*}}  \mathds{1}_{\{ q^{m}_{s-}+z \in \mathcal{Q} \}}  \Lambda^{b}(\delta^{b}(s,z)) \left( \varphi(s,q^{m}_{s-}+z) - \varphi(s,q^{m}_{s-}) \right) \mu^{b}(dz) ds\nonumber\\
\end{split}
\end{equation}
\begin{equation}
\begin{split}
& + \int\limits_{t_{m}}^{\pi_{m}}  \int_{\mathbb{R}_{+}^{*}} \mathds{1}_{\{ q^{m}_{s-}-z\in \mathcal{Q} \}} \Lambda^{a}(\delta^{a}(s,z)) \left( \varphi(s,q^{m}_{s-}-z) - \varphi(s,q^{m}_{s-}) \right) \mu^{a}(dz) ds\\
& + \int\limits_{t_{m}}^{\pi_{m}}  \int_{\mathbb{R}_{+}^{*}} \left( \varphi(s,q^{m}_{s-}+z) - \varphi(s,q^{m}_{s-}) \right) \tilde{J}^{b}(ds,dz)\\
& + \int\limits_{t_{m}}^{\pi_{m}}  \int_{\mathbb{R}_{+}^{*}} \left( \varphi(s,q^{m}_{s-}-z) - \varphi(s,q^{m}_{s-}) \right) \tilde{J}^{a}(ds,dz)\\
& + \int\limits_{t_{m}}^{\pi_{m}}  \partial_{q} \varphi(s,q^{m}_{s-}) \left(-(v_s)_- \zeta(\tilde q + q^{m}_{s-}) + (v_s)_+ \zeta(\tilde q - q^{m}_{s-}) \right) ds,\nonumber
\end{split}
\end{equation}
which we can write
\begin{equation}
\begin{split}
\varphi( & \pi_{m}, q^{m}_{\pi_{m}}) = \varphi(t_{m},q_{m}) + \int_{t_{m}}^{\pi_{m}} \bigg\lbrace \frac{\partial \varphi}{\partial t}(s,q^{m}_{s-})\\
& +  \int_{\mathbb{R}_{+}^{*}}  \mathds{1}_{\{ q^{m}_{s-}+z \in \mathcal{Q} \}} z\Lambda^{b}(\delta^{b}(s,z)) \left(\delta^{b}(s,z) - \frac{ \varphi(s,q^{m}_{s-}) - \varphi(s,q^{m}_{s-}+z)}{z} \right) \mu^{b}(dz)\\
& + \int_{\mathbb{R}_{+}^{*}}  \mathds{1}_{\{ q^{m}_{s-}-z \in \mathcal{Q} \}} z \Lambda^{a}(\delta^{a}(s,z)) \left(\delta^{a}(s,z) - \frac{ \varphi(s,q^{m}_{s-}) - \varphi(s,q^{m}_{s-}-z)}{z} \right) \mu^{a}(dz)\\
& +   \left(-(v_s)_- \zeta(\tilde q + q^{m}_{s-}) + (v_s)_+ \zeta(\tilde q - q^{m}_{s-}) \right) \left( \partial_{q} \varphi(s,q^{m}_{s-}) + kq^{m}_{s-} \right) - \tilde{\mathcal L}(v_s, q^{m}_{s-})   - \psi(q^{m}_{s-}) \bigg\rbrace ds\\
& + \int\limits_{t_{m}}^{\pi_{m}} \bigg\lbrace \psi(q^{m}_{s-}) - k\left(-(v_s)_- \zeta(\tilde q + q^{m}_{s-}) + (v_s)_+ \zeta(\tilde q - q^{m}_{s-}) \right) q^{m}_{s-} -  \int_{\mathbb{R}_{+}^{*}}  \mathds{1}_{\{ q^{m}_{s-}+z \in \mathcal{Q} \}}  z\Lambda^{b}(\delta^{b}(s,z)) \delta^{b}(s,z) \mu^{b}(dz)\\
& -  \int_{\mathbb{R}_{+}^{*}}  \mathds{1}_{\{ q^{m}_{s-}-z \in \mathcal{Q} \}} z \Lambda^{a}(\delta^{a}(s,z)) \delta^{a}(s,z) \mu^{a}(dz) + \tilde{\mathcal L}(v_s, q^{m}_{s-}) \bigg\rbrace ds\\
& + \int\limits_{t_{m}}^{\pi_{m}}  \int_{\mathbb{R}_{+}^{*}} \left( \varphi(s,q^{m}_{s-}+z) - \varphi(s,q^{m}_{s-}) \right) \tilde{J}^{b}(ds,dz)\\
& + \int\limits_{t_{m}}^{\pi_{m}}  \int_{\mathbb{R}_{+}^{*}} \left( \varphi(s,q^{m}_{s-}-z) - \varphi(s,q^{m}_{s-}) \right) \tilde{J}^{a}(ds,dz).\nonumber
\end{split}
\end{equation}

From \eqref{eq2}, and by definition of $H^{b}$, $H^{a}$, and $\mathcal H$, we then get
\begin{equation}
\begin{split}
\varphi(\pi_{m}, q^{m}_{\pi_{m}}) \leq & \  \varphi(t_{m},q_{m}) + \int\limits_{t_{m}}^{\pi_{m}} \bigg\lbrace \psi(q^{m}_{s-}) - k\left(-(v_s)_- \zeta(\tilde q + q^{m}_{s-}) + (v_s)_+ \zeta(\tilde q - q^{m}_{s-}) \right)q^{m}_{s-}\\
& -  \int_{\mathbb{R}_{+}^{*}}  \mathds{1}_{\{ q^{m}_{s-}+z\in \mathcal{Q} \}} z\Lambda^{b}(\delta^{b}(s,z)) \delta^{b}(s,z) \mu^{b}(dz)\\
& - \int_{\mathbb{R}_{+}^{*}}  \mathds{1}_{\{ q^{m}_{s-}-z \in \mathcal{Q} \}} z \Lambda^{a}(\delta^{a}(s,z)) \delta^{a}(s,z) \mu^{a}(dz) + \tilde{\mathcal L}(v_s, q^{m}_{s-}) \bigg\rbrace ds\\
& + \int\limits_{t_{m}}^{\pi_{m}} \int_{\mathbb{R}_{+}^{*}} \left( \varphi(s,q^{m}_{s-}+z) - \varphi(s,q^{m}_{s-}) \right) \tilde{J}^{b}(ds,dz)\\
& + \int\limits_{t_{m}}^{\pi_{m}}  \int_{\mathbb{R}_{+}^{*}} \left( \varphi(s,q^{m}_{s-}-z) - \varphi(s,q^{m}_{s-}) \right) \tilde{J}^{a}(ds,dz).\nonumber
\end{split}
\end{equation}

The last two terms have expectations equal to zero and we obtain
\begin{equation}
\begin{split}
\varphi(t_{m},q_{m}) \geq  \mathbb{E} \Bigg[ &\ \varphi(\pi_{m}, q^{m}_{\pi_{m}}) + \int\limits_{t_{m}}^{\pi_{m}} \bigg\lbrace  \int_{\mathbb{R}_{+}^{*}}  \mathds{1}_{\{ q^{m}_{s-}+z \in \mathcal{Q} \}} z\Lambda^{b}(\delta^{b}(s,z)) \delta^{b}(s,z) \mu^{b}(dz)\\
& + \int_{\mathbb{R}_{+}^{*}} \mathds{1}_{\{ q^{m}_{s-}-z \in \mathcal{Q} \}}  z \Lambda^{a}(\delta^{a}(s,z)) \delta^{a}(s,z) \mu^{a}(dz)\\
& + k q^{m}_{s-} \left(-(v_s)_- \zeta(\tilde q + q^{m}_{s-}) + (v_s)_+ \zeta(\tilde q - q^{m}_{s-}) \right) - \tilde{\mathcal L}(v_s, q^{m}_{s-}) - \psi(q^{m}_{s-}) \bigg\rbrace ds \Bigg]. \nonumber
\end{split}
\end{equation}

Therefore
\begin{equation}
\begin{split}
\varphi(t_{m},q_{m}) \geq    \eta +  & \mathbb{E} \Bigg[ \theta(\pi_{m}, q^{m}_{\pi_{m}}) + \int\limits_{t_{m}}^{\pi_{m}} \bigg\lbrace  \int_{\mathbb{R}_{+}^{*}}  \mathds{1}_{\{ q^{m}_{s-}+z \in \mathcal{Q} \}} z\Lambda^{b}(\delta^{b}(s,z)) \delta^{b}(s,z) \mu^{b}(dz)\\
& +\int_{\mathbb{R}_{+}^{*}}  \mathds{1}_{\{ q^{m}_{s-}-z \in \mathcal{Q} \}} z \Lambda^{a}(\delta^{a}(s,z)) \delta^{a}(s,z) \mu^{a}(dz)\\
& +  k q^{m}_{s-} \left(-(v_s)_- \zeta(\tilde q + q^{m}_{s-}) + (v_s)_+ \zeta(\tilde q - q^{m}_{s-}) \right) - \tilde{\mathcal L}(v_s, q^{m}_{s-}) - \psi(q^{m}_{s-}) \bigg\rbrace ds \Bigg]. \nonumber
\end{split}
\end{equation}

As $\varphi(t_{m},q_{m}) \xrightarrow[m\rightarrow +\infty]{} \varphi(\bar{t},\bar{q}) = \theta^{*}(\bar{t},\bar{q})$ and $\theta(t_{m},q_{m})\xrightarrow[m\rightarrow +\infty]{}\theta^{*}(\bar{t},\bar{q})$, we have for $m$ large enough the inequality $\theta(t_{m},q_{m}) + \frac{\eta}{2} \geq \varphi(t_{m},q_{m})$, from  which we deduce
\begin{equation}
\begin{split}
\theta(t_{m},q_{m}) \geq    \frac{\eta}{2} + & \mathbb{E} \Bigg[ \theta(\pi_{m}, q^{m}_{\pi_{m}}) + \int\limits_{t_{m}}^{\pi_{m}} \bigg\lbrace  \int_{\mathbb{R}_{+}^{*}}  \mathds{1}_{\{ q^{m}_{s-}+z \in \mathcal{Q} \}} z\Lambda^{b}(\delta^{b}(s,z)) \delta^{b}(s,z) \mu^{b}(dz)\\
& + \int_{\mathbb{R}_{+}^{*}}  \mathds{1}_{\{ q^{m}_{s-}-z \in \mathcal{Q} \}} z \Lambda^{a}(\delta^{a}(s,z)) \delta^{a}(s,z) \mu^{a}(dz)\\
& + k q^{m}_{s-} \left(-(v_s)_- \zeta(\tilde q + q^{m}_{s-}) + (v_s)_+ \zeta(\tilde q - q^{m}_{s-}) \right) - \tilde{\mathcal L}(v_s, q^{m}_{s-}) - \psi(q^{m}_{s-}) \bigg\rbrace ds \Bigg]. \nonumber
\end{split}
\end{equation}

By taking the supremum over all the controls in $\mathcal{A}$ on the right-hand side, we contradict the dynamic programming principle. \newline

Necessarily, we deduce:
\begin{equation}
\begin{split}
-\frac{\partial \varphi}{\partial t}(\bar{t}, \bar{q}) &+ \psi(\bar{q}) - \int_{\mathbb{R}_{+}^{*}}  \mathds{1}_{\{ \bar{q}+z\in \mathcal{Q} \}} zH^{b} \left(\frac{\varphi(\bar{t}, \bar{q}) -  \varphi(\bar{t}, \bar{q}+z) }{z}\right) \mu^{b}(dz)  \\
&-  \int_{\mathbb{R}_{+}^{*}}  \mathds{1}_{\{ \bar{q}-z \in \mathcal{Q} \}}zH^{a} \left( \frac{\varphi(\bar{t}, \bar{q}) - \varphi(\bar{t}, \bar{q}-z)}{z} \right) \mu^{a}(dz) -  \mathcal H \left(\partial_{q} \varphi(\bar{t},\bar q),\bar{q}\right) \ \leq 0, \nonumber
\end{split}
\end{equation}
and $\theta$ is a viscosity subsolution to $\eqref{eqn:HJB}$ on $[0,T) \times \mathcal Q$.\\
\end{proof}

\begin{prop}
$\theta$ is a viscosity supersolution to \eqref{eqn:HJB} on $[0,T) \times \mathcal Q $.\newline
\end{prop}

\begin{proof}
$\theta$ is bounded on $[0,T]\times \mathcal Q$, so we can define $\theta_{*}$ its lower semicontinuous envelope.\\

Let $(\bar{t},\bar{q}) \in [0,T) \times \mathcal Q$ and $\varphi \in C^1$ such that
\begin{equation}
0 = (\theta_{*} - \varphi)(\bar{t},\bar{q}) = \underset{(t,q)\in[0,T) \times \mathcal Q}{\min}\ (\theta_{*} - \varphi)(t,q). \nonumber
\end{equation}

We can assume this minimum to be strict. By definition of $\theta_{*}(\bar{t},\bar{q})$, there exists $(t_{m},q_{m})_{m}$ a sequence of $[0,T) \times \mathcal Q$ such that
\begin{equation}
\begin{split}
(t_{m},q_{m}) \xrightarrow[{m \rightarrow +\infty}]{} (\bar{t},\bar{q}),\\
\theta(t_{m},q_{m}) \xrightarrow[{m \rightarrow +\infty}]{} \theta_{*}(\bar{t},\bar{q}). \nonumber
\end{split}
\end{equation}

Let us prove the proposition by contradiction. Assume there is $\eta > 0$ such that
\begin{equation}
\begin{split}
-\frac{\partial \varphi}{\partial t}(\bar{t}, \bar{q}) &+ \psi(\bar{q}) -  \int_{\mathbb{R}_{+}^{*}} \mathds{1}_{\{ \bar{q}+z \in \mathcal{Q} \}} zH^{b} \left(\frac{\varphi(\bar{t}, \bar{q}) -  \varphi(\bar{t}, \bar{q}+z) }{z}\right) \mu^{b}(dz)  \\
&-  \int_{\mathbb{R}_{+}^{*}} \mathds{1}_{\{ \bar{q}-z \in \mathcal{Q} \}} zH^{a} \left( \frac{\varphi(\bar{t}, \bar{q}) - \varphi(\bar{t}, \bar{q}-z)}{z} \right) \mu^{a}(dz) - \mathcal H \left(\partial_{q} \varphi(\bar{t},\bar q),\bar{q} \right) \ < - 2\eta. \nonumber
\end{split}
\end{equation}

Then, as $\varphi$ is continuously differentiable and $\mu^{b}$ and $\mu^{a}$ absolutely continuous with respect to the Lebesgue measure, we must have
\begin{equation}
\label{eq3}
\begin{split}
-\frac{\partial \varphi}{\partial t}(t,q) &+ \psi(q) -  \int_{\mathbb{R}_{+}^{*}} \mathds{1}_{\{ q+z \in \mathcal{Q} \}} zH^{b} \left(\frac{\varphi(t,q) -  \varphi(t,q+z) }{z}\right) \mu^{b}(dz)  \\
&- \int_{\mathbb{R}_{+}^{*}} \mathds{1}_{\{ q-z \in \mathcal{Q} \}} zH^{a} \left(\frac{\varphi(t,q) - \varphi(t,q-z)}{z} \right) \mu^{a}(dz)  - \mathcal H \left(\partial_{q} \varphi(t,q),q \right) \ \leq 0
\end{split}
\end{equation}
on $B:=\big((\bar{t}-r,\bar{t}+r)\cap [0,T) \big) \times \left( (\bar{q}-r,\bar q + r) \cap \mathcal Q \right)  $ for a sufficiently small $r \in \big(0,T-\bar{t} \big)$. Without loss of generality, we can assume that $B$ contains the sequence $(t_{m},q_{m})_{m}$.\newline

Then, by potentially reducing $\eta$, we have
\begin{equation}
\theta \geq \theta_{*} \geq \varphi + \eta \nonumber
\end{equation}
on $\partial_{p}B$. We can also without loss of generality assume that this inequality is true on
\begin{equation}
\tilde{B}:= \{ (t,q+z)\ |\ (t,q,z) \in B \times \mathbb R,\  q+z \in (\bar{q}-r,\bar q + r)^{c} \cap \mathcal Q \}, \nonumber
\end{equation}
which is also bounded.\newline

We introduce the control $\delta=(\delta^{b},\delta^{a})\in \mathcal{A}_{\mathcal M}$ such that $\forall t\geq t_{m}$, $\forall z \in \mathbb{R}_{+}^{*}$,
$$\delta^{b}(t,z) = \delta^{b*}\left(\frac{\varphi(t,q^{m}_{t-}) -  \varphi(t,q^{m}_{t-}+z) }{z}\right) \quad \textrm{and} \quad \delta^{a}(t,z) = \delta^{a*}\left( \frac{\varphi(t,q^{m}_{t-}) - \varphi(t,q^{m}_{t-}-z)}{z} \right),$$ where $\delta^{b*}$ and $\delta^{a*}$ are defined in Lemma~\ref{lemmH}. Similarly, we introduce the control $ v \in \mathcal{A}_{\mathcal T}$ such that $\forall t\geq t_m,$ $$v_t= v^{*} \left( \partial_{q}\varphi(t,q^m_t), q_t^m \right),$$ where $v^{*}$ is defined in Lemma~\ref{lemmHronde}. As before, we denote by $\pi_{m}$ the first exit time of $(t,q^{m}_{t})_{t\geq t_{m}}$ from $B$ (where $q^{m}_{t}:=q^{t_{m},q_{m}, {\delta}, {v}}_{t}$). By Itô's lemma, we obtain
\begin{equation}
\begin{split}
\varphi(\pi_{m}, q^{m}_{\pi_{m}}) & = \varphi(t_{m},q_{m}) + \int\limits_{t_{m}}^{\pi_{m}} \frac{\partial \varphi}{\partial t}(s,q^{m}_{s-}) ds\\
& + \int\limits_{t_{m}}^{\pi_{m}}  \int_{\mathbb{R}_{+}^{*}} \mathds{1}_{\{ q^m_{s-} +z \in \mathcal{Q} \}}  \Lambda^{b}(\delta^{b}(s,z)) \left( \varphi(s,q^{m}_{s-}+z) - \varphi(s,q^{m}_{s-}) \right) \mu^{b}(dz) ds\nonumber\\
& + \int\limits_{t_{m}}^{\pi_{m}} \int_{\mathbb{R}_{+}^{*}} \mathds{1}_{\{ q^{m}_{s-}-z \in \mathcal{Q} \}} \Lambda^{a}(\delta^{a}(s,z)) \left( \varphi(s,q^{m}_{s-}-z) - \varphi(s,q^{m}_{s-}) \right) \mu^{a}(dz) ds\\
\end{split}
\end{equation}
\begin{equation}
\begin{split}
& + \int\limits_{t_{m}}^{\pi_{m}} \int_{\mathbb{R}_{+}^{*}} \left( \varphi(s,q^{m}_{s-}+z) - \varphi(s,q^{m}_{s-}) \right) \tilde{J}^{b}(ds,dz)\\
& + \int\limits_{t_{m}}^{\pi_{m}} \int_{\mathbb{R}_{+}^{*}} \left( \varphi(s,q^{m}_{s-}-z) - \varphi(s,q^{m}_{s-}) \right) \tilde{J}^{a}(ds,dz)\\
& + \int\limits_{t_{m}}^{\pi_{m}}  \partial_{q} \varphi(s,q^{m}_{s-}) \left(-(v_s)_- \zeta(\tilde q + q^{m}_{s-}) + (v_s)_+ \zeta(\tilde q - q^{m}_{s-}) \right) ds,\nonumber
\end{split}
\end{equation}
which we can write
\begin{equation}
\begin{split}
\varphi( & \pi_{m}, q^{m}_{\pi_{m}}) = \varphi(t_{m},q_{m}) + \int_{t_{m}}^{\pi_{m}} \bigg\lbrace \frac{\partial \varphi}{\partial t}(s,q^{m}_{s-})\\
& + \int_{\mathbb{R}_{+}^{*}} \mathds{1}_{\{ q^m_{s-} +z \in \mathcal{Q} \}} z\Lambda^{b}(\delta^{b}(s,z)) \left(\delta^{b}(s,z) - \frac{ \varphi(s,q^{m}_{s-}) - \varphi(s,q^{m}_{s-}+z)}{z} \right) \mu^{b}(dz)\\
& + \int_{\mathbb{R}_{+}^{*}} \mathds{1}_{\{ q^m_{s-} -z \in \mathcal{Q} \}} z \Lambda^{a}(\delta^{a}(s,z)) \left(\delta^{a}(s,z) - \frac{ \varphi(s,q^{m}_{s-}) - \varphi(s,q^{m}_{s-}-z)}{z} \right) \mu^{a}(dz)\\
& +  \left(-(v_s)_- \zeta(\tilde q + q^{m}_{s-}) + (v_s)_+ \zeta(\tilde q - q^{m}_{s-}) \right) \left(\partial_{q} \varphi(s,q^{m}_{s-}) + k q^{m}_{s-} \right) - \tilde{\mathcal L}(v_s, q^{m}_{s-})   - \psi(q^{m}_{s-}) \bigg\rbrace ds\\
& + \int\limits_{t_{m}}^{\pi_{m}} \bigg\lbrace \psi(q^{m}_{s-}) -  kq^{m}_{s-} \left(-(v_s)_- \zeta(\tilde q + q^{m}_{s-}) + (v_s)_+ \zeta(\tilde q - q^{m}_{s-}) \right)\\
& - \int_{\mathbb{R}_{+}^{*}} \mathds{1}_{\{ q^m_{s-} +z \in \mathcal{Q} \}} z\Lambda^{b}(\delta^{b}(s,z)) \delta^{b}(s,z) \mu^{b}(dz)\\
& - \int_{\mathbb{R}_{+}^{*}} \mathds{1}_{\{ q^m_{s-} -z\in \mathcal{Q} \}} z \Lambda^{a}(\delta^{a}(s,z)) \delta^{a}(s,z) \mu^{a}(dz) +\tilde{\mathcal L}(v_s, q^{m}_{s-}) \bigg\rbrace ds\\
& + \int\limits_{t_{m}}^{\pi_{m}}  \int_{\mathbb{R}_{+}^{*}} \left( \varphi(s,q^{m}_{s-}+z) - \varphi(s,q^{m}_{s-}) \right) \tilde{J}^{b}(ds,dz)\\
& + \int\limits_{t_{m}}^{\pi_{m}}  \int_{\mathbb{R}_{+}^{*}} \left( \varphi(s,q^{m}_{s-}-z) - \varphi(s,q^{m}_{s-}) \right) \tilde{J}^{a}(ds,dz).\nonumber
\end{split}
\end{equation}

By \eqref{eq3}, we then get
\begin{equation}
\begin{split}
\varphi(\pi_{m}, q^{m}_{\pi_{m}}) \geq & \  \varphi(t_{m},q_{m}) + \int\limits_{t_{m}}^{\pi_{m}} \bigg\lbrace \psi(q^{m}_{s-}) -  k q^{m}_{s-} \left(-(v_s)_- \zeta(\tilde q + q^{m}_{s-}) + (v_s)_+ \zeta(\tilde q - q^{m}_{s-}) \right) \\
& - \int_{\mathbb{R}_{+}^{*}} \mathds{1}_{\{ q^m_{s-} +z \in \mathcal{Q} \}} z\Lambda^{b}(\delta^{b}(s,z)) \delta^{b}(s,z) \mu^{b}(dz)\\
& - \int_{\mathbb{R}_{+}^{*}} \mathds{1}_{\{ q^m_{s-} -z \in \mathcal{Q} \}} z \Lambda^{a}(\delta^{a}(s,z)) \delta^{a}(s,z) \mu^{a}(dz) +\tilde{\mathcal L}(v_s, q^{m}_{s-}) \bigg\rbrace ds\\
& + \int\limits_{t_{m}}^{\pi_{m}}  \int_{\mathbb{R}_{+}^{*}} \left( \varphi(s,q^{m}_{s-}+z) - \varphi(s,q^{m}_{s-}) \right) \tilde{J}^{b}(ds,dz)\\
& + \int\limits_{t_{m}}^{\pi_{m}}  \int_{\mathbb{R}_{+}^{*}} \left( \varphi(s,q^{m}_{s-}-z) - \varphi(s,q^{m}_{s-}) \right) \tilde{J}^{a}(ds,dz).\nonumber
\end{split}
\end{equation}

The last two terms have expectations equal to zero and we obtain
\begin{equation}
\begin{split}
\varphi(t_{m},q_{m}) \leq  \mathbb{E} \Bigg[ & \varphi(\pi_{m}, q^{m}_{\pi_{m}}) + \int\limits_{t_{m}}^{\pi_{m}} \bigg\lbrace  \int_{\mathbb{R}_{+}^{*}} \mathds{1}_{\{ q^m_{s-} +z\in \mathcal{Q} \}} z\Lambda^{b}(\delta^{b}(s,z)) \delta^{b}(s,z) \mu^{b}(dz)\\
& + \int_{\mathbb{R}_{+}^{*}} \mathds{1}_{\{ q^m_{s-} -z \in \mathcal{Q} \}} z \Lambda^{a}(\delta^{a}(s,z)) \delta^{a}(s,z) \mu^{a}(dz)\\
& + k q^{m}_{s-} \left(-(v_s)_- \zeta(\tilde q + q^{m}_{s-}) + (v_s)_+ \zeta(\tilde q - q^{m}_{s-}) \right) - \tilde{\mathcal L}(v_s, q^{m}_{s-}) - \psi(q^{m}_{s-}) \bigg\rbrace ds \Bigg]. \nonumber
\end{split}
\end{equation}

Therefore,
\begin{equation}
\begin{split}
\varphi(t_{m},q_{m}) \leq - \eta + \mathbb{E} \Bigg[ & \theta(\pi_{m}, q^{m}_{\pi_{m}}) + \int\limits_{t_{m}}^{\pi_{m}} \bigg\lbrace \int_{\mathbb{R}_{+}^{*}} \mathds{1}_{\{ q^m_{s-} +z\in \mathcal{Q} \}}  z\Lambda^{b}(\delta^{b}(s,z)) \delta^{b}(s,z) \mu^{b}(dz)\\
& +  \int_{\mathbb{R}_{+}^{*}} \mathds{1}_{\{ q^m_{s-} -z \in \mathcal{Q} \}} z \Lambda^{a}(\delta^{a}(s,z)) \delta^{a}(s,z) \mu^{a}(dz)\\
& + k q^{m}_{s-} \left(-(v_s)_- \zeta(\tilde q + q^{m}_{s-}) + (v_s)_+ \zeta(\tilde q - q^{m}_{s-}) \right) - \tilde{\mathcal L}(v_s, q^{m}_{s-}) - \psi(q^{m}_{s-}) \bigg\rbrace ds \Bigg]. \nonumber
\end{split}
\end{equation}

As $\varphi(t_{m},q_{m}) \xrightarrow[m\rightarrow +\infty]{} \varphi(\bar{t},\bar{q}) = \theta_{*}(\bar{t},\bar{q})$ and moreover $\theta(t_{m},q_{m})\xrightarrow[m\rightarrow +\infty]{}\theta_{*}(\bar{t},\bar{q})$, we have that for $m$ sufficiently large, $\theta(t_{m},q_{m}) - \frac{\eta}{2} \leq \varphi(t_{m},q_{m})$ and we deduce:
\begin{equation}
\begin{split}
\theta(t_{m},q_{m}) <   \mathbb{E} \Bigg[ &\theta(\pi_{m}, q^{m}_{\pi_{m}}) + \int\limits_{t_{m}}^{\pi_{m}} \bigg\lbrace  \int_{\mathbb{R}_{+}^{*}} \mathds{1}_{\{ q^m_{s-} +z \in \mathcal{Q} \}} z\Lambda^{b}(\delta^{b}(s,z)) \delta^{b}(s,z) \mu^{b}(dz)\\
& +  \int_{\mathbb{R}_{+}^{*}} \mathds{1}_{\{ q^m_{s-} -z \in \mathcal{Q} \}} z \Lambda^{a}(\delta^{a}(s,z)) \delta^{a}(s,z) \mu^{a}(dz)\\
& + k q^{m}_{s-} \left(-(v_s)_- \zeta(\tilde q + q^{m}_{s-}) + (v_s)_+ \zeta(\tilde q - q^{m}_{s-}) \right) - \tilde{\mathcal L}(v_s, q^{m}_{s-}) - \psi(q^{m}_{s-}) \bigg\rbrace ds \Bigg], \nonumber
\end{split}
\end{equation}
which contradicts the dynamic programming principle. \newline

In conclusion, we necessarily have
\begin{equation}
\begin{split}
-\frac{\partial \varphi}{\partial t}(\bar{t}, \bar{q}) &+ \psi(\bar{q}) -  \int_{\mathbb{R}_{+}^{*}}\mathds{1}_{\{ \bar{q} +z \in \mathcal{Q} \}} zH^{b} \left(\frac{\varphi(\bar{t}, \bar{q}) -  \varphi(\bar{t}, \bar{q}+z) }{z}\right) \mu^{b}(dz)  \\
&-  \int_{\mathbb{R}_{+}^{*}}\mathds{1}_{\{ \bar{q}-z \in \mathcal{Q} \}} zH^{a} \left( \frac{\varphi(\bar{t}, \bar{q}) - \varphi(\bar{t}, \bar{q}-z)}{z} \right) \mu^{a}(dz) - \mathcal H \left(\partial_{q} \varphi(\bar{t},\bar q),\bar{q} \right) \ \geq 0, \nonumber
\end{split}
\end{equation}
and $\theta$ is a viscosity supersolution to $\eqref{eqn:HJB}$ on $[0,T) \times \mathcal Q$.\newline

\end{proof}

\begin{prop}
$\forall q \in \mathcal Q,$ we have $\theta_*(T,q) = \theta^*(T,q)= -\ell(q)$.
\end{prop}

\begin{proof}
Let $q\in \mathcal Q$ and let us take $(t_m,q_m)_{m\in \mathbb{N}}$ a sequence of $[0,T] \times \mathcal Q$ such that $$(t_m,q_m) \underset{m\rightarrow +\infty}{\longrightarrow }(T,q) \qquad \text{and} \qquad \theta(t_m,q_m) \underset{m\rightarrow +\infty}{\longrightarrow}\theta_*(T,q).$$ 
We introduce arbitrary controls $ \delta=(\delta^{b},\delta^{a})\in \mathcal{A}_{\mathcal M}$. We also introduce an arbitrary control $ v \in \mathcal{A}_{\mathcal T}$. Then we have for all $m\in \mathbb{N},$ by denoting $q^m_t = q^{t_m, q_m, {\delta}, {v}}_t$ for all $t\in [t_m,T]$:
\begin{equation*}
\begin{split}
\theta(t_m,q_m) \geq \mathbb{E} & \left[\int\limits_{t_m}^{T} \Bigg\lbrace \int_{\mathbb{R}_{+}^{*}} \Big(z\delta^{i,b}(s,z) \mathds{1}_{\left\{q^m_{s-} + z \in \mathcal Q \right\}} \Lambda^{b}(\delta^{b}(s,z))\mu^{b}(dz) \right.\\
&+ z\delta^{a}(s,z) \mathds{1}_{\left\{q^m_{s-} - z\in \mathcal Q \right\}} \Lambda^{a}(\delta^{a}(s,z))\mu^{a}(dz) \Big) \\
&\left.+ kq^{m}_{s-} \left(-(v_s)_- \zeta(\tilde q + q^{m}_{s-}) + (v_s)_+ \zeta(\tilde q - q^{m}_{s-}) \right) - \tilde{\mathcal L}(v_s, q^{m}_{s-}) - \psi(q^m_{s-}) \Bigg\rbrace ds -\ell(q^m_{T-}) \vphantom{\int\limits_{t}^{T}} \right].  
\end{split}    
\end{equation*}
But, 
$$
\Bigg| \int\limits_{t_m}^{T} \Bigg\lbrace \int_{\mathbb{R}_{+}^{*}} \Big(z\delta^{b}(s,z) \mathds{1}_{\left\{q^m_{s-} + z \in \mathcal Q \right\}} \Lambda^{b}(\delta^{b}(s,z))\mu^{b}(dz)+ z\delta^{a}(s,z)\mathds{1}_{\left\{q^m_{s-} - z \in \mathcal Q \right\}} \Lambda^{a}(\delta^{a}(s,z))\mu^{a}(dz) \Big) \Bigg\rbrace ds \Bigg|$$\vspace{-3mm}$$\leq (T-t_m) \left( \Delta^{b} H^b(0)+ \Delta^{a} H^a(0) \right),$$
$$\left| \int_{t_m}^T k q^{m}_{s-} \left(-(v_s)_- \zeta(\tilde q + q^{m}_{s-}) + (v_s)_+ \zeta(\tilde q - q^{m}_{s-}) \right)  ds \right| \le 2(T-t_m) k v_{\infty} \tilde q,$$
$$
 \left|\int_{t_m}^T \tilde{\mathcal L}(v_s, q^{m}_{s-}) ds \right| \leq ( T-t_m) \left( L(-v_{\infty}) + L(v_{\infty}) \right)$$
and $$\left|\int_{t_m}^T \psi(q^{m}_{s-}) ds\right|\le (T-t_m) \sup_{q \in Q} \psi(q).$$

By dominated convergence (because $(q^m_T)_m$ converges in probability towards $q$ and $\ell$ is continuous on the compact set $\mathcal{Q}$), we have $\mathbb{E} \left[\ell(q^m_T) \right] \underset{m\rightarrow +\infty}{\longrightarrow}\ell(q)$, and therefore $\theta_*(T,q) \geq -\ell(q).$ But as $\theta_*(T,q) \leq \theta(T,q) = -\ell(q),$ we get $\theta_*(T,q) = -\ell(q).$\\

The proof for $\theta^*$ is similar, by taking $\varepsilon$-optimal controls and showing that $\theta^*(T,q) - \varepsilon \leq -\ell(q)$ for all $\varepsilon>0.$
\end{proof}

\subsection{Uniqueness result}

\begin{thm}
\label{unic}
Let $u$ be a bounded USC subsolution and $v$ be a bounded LSC supersolution to \eqref{eqn:HJB} on $[0,T) \times \mathcal Q$ such that $u\leq v$ on $\{T\} \times \mathcal Q$. Then $u\leq v$ on $[0,T] \times \mathcal Q.$
\end{thm}

\begin{proof}
We prove it by contradiction. Let us assume $\underset{[0,T] \times \mathcal Q}{\sup} ( u-v ) >0$. Then this supremum cannot be reached on $\{T\} \times \mathcal Q$. For $n \geq 0$ and $\varepsilon> 0$, we introduce:

\begin{equation}
\phi_{n,\varepsilon}(t,s,q,y) = u(t,q) - v(s,y) - n(q-y)^{2} - n(t-s)^{2} - \varepsilon(2T-t-s). \nonumber
\end{equation}

We also introduce $(t_{n,\varepsilon},s_{n,\varepsilon},q_{n,\varepsilon},y_{n,\varepsilon})$ such that:

\begin{equation}
\phi_{n,\varepsilon}(t_{n,\varepsilon},s_{n,\varepsilon},q_{n,\varepsilon},y_{n,\varepsilon}) = \underset{[0,T]^{2} \times \mathcal Q^{2}}{\max}  \phi_{n,\varepsilon}(t,s,q,y).  \nonumber
\end{equation}

Then for all $n\geq 0, \epsilon>0$ and for all $(t,q) \in [0,T] \times \mathcal Q$, we have $$\phi_{n,\varepsilon}(t_{n,\varepsilon},s_{n,\varepsilon},q_{n,\varepsilon},y_{n,\varepsilon}) \geq \phi_{n,\varepsilon}(t,t,q,q) =  u(t,q)-v(t,q) - 2\varepsilon(T-t).$$ In particular, 
\begin{equation}
\label{posIn}
\phi_{n,\varepsilon}(t_{n,\varepsilon},s_{n,\varepsilon},q_{n,\varepsilon},y_{n,\varepsilon}) \geq \underset{[0,T] \times \mathcal Q}{\sup} \left( u-v\right) - 2\varepsilon T.   
\end{equation}
We can now fix $\varepsilon$ such that $$0<\varepsilon<\frac{\underset{[0,T] \times \mathcal Q}{\sup}\left(u - v  \right)}{4T}$$ to ensure that the right-hand side of $\eqref{posIn}$ is larger than $\frac 12 \underset{[0,T] \times \mathcal Q}{\sup} \left( u-v \right)$, which is positive by assumption. $\varepsilon$ will remain fixed throughout the rest of the proof and, for ease of notation, we now write $\phi_n = \phi_{n,\varepsilon}$, $t_n = t_{n,\varepsilon}$ , $s_n = s_{n,\varepsilon}$, $q_n = q_{n,\varepsilon}$, and $y_n = y_{n,\varepsilon}$.\\

From what precedes, we know that the sequence $\left(n(q_n - y_n)^2 + n(t_n-s_n)^2\right)_n$ is bounded, so necessarily, $|t_n - s_n| \underset{n\rightarrow +\infty}{\longrightarrow}$0 and $|q_n - y_n| \underset{n\rightarrow +\infty}{
\longrightarrow}0$. Then, up to a subsequence, there exists $(\bar{t},\bar{q}) \in [0,T] \times \mathcal{Q}$ such that $s_n, t_n \underset{n\rightarrow +\infty}{\longrightarrow}\bar{t}$ and $q_n, y_n \underset{n\rightarrow +\infty}{\longrightarrow}\bar{q}$.\\

Moreover, we know that
$$\phi_n(t_n,s_n,q_n,y_n) \ge \phi_n(\bar t,\bar t,\bar q,\bar q),$$
which implies
\begin{align*}
    u(t_n,q_n) - v(s_n,y_n) - n(q_n-y_n)^{2} - n(t_n-s_n)^{2} - \varepsilon(2T-t_n-s_n) \ge u(\bar t, \bar q) - v(\bar t, \bar q) -2\varepsilon(T-\bar t).
\end{align*}
Hence we have
\begin{align*}
    n(q_n-y_n)^{2} + n(t_n-s_n)^{2} \le u(t_n,q_n) - u(\bar t, \bar q) +v(\bar t, \bar q) - v(s_n,y_n) +2\varepsilon(T-\bar t) - \varepsilon(2T-t_n-s_n) .
\end{align*}
As $u$ is USC and $v$ is LSC, the $\limsup$ when $n\rightarrow +\infty$ of the left-hand side is nonpositive, which implies $ n(q_n-y_n)^{2} + n(t_n-s_n)^{2}\xrightarrow[n \rightarrow +\infty]{} 0 $.\\

Let us assume $\bar{t} = T$. Then we have, as $u$ is USC and $v$ is LSC:
\begin{equation}
    \underset{n \rightarrow + \infty}{\limsup}\ \phi_n(t_n,s_n,q_n,y_n) \leq \underset{n \rightarrow + \infty}{\limsup}\ u(t_n,q_n) - \underset{n \rightarrow + \infty}{\liminf}\ v(s_n,y_n) \leq u(T,\bar{q}) - v(T,\bar{q}) \leq 0,\nonumber
\end{equation}
which constitutes a contradiction. Necessarily, $\bar{t}<T$.\\

Hence, for $n$ large enough we must have $t_{n},s_{n}<T$, and we know that $(t_{n},q_{n})$ is a maximum point of $u-\varphi_{n}$ where $$\varphi_{n}(t,q) = v(s_{n},y_{n}) + n(q-y_{n})^{2} + n(t-s_{n})^{2} + \varepsilon(2T-t-s_{n}).$$ By Proposition~\ref{eqvisco}, we have:
\begin{equation}
\begin{split}
\varepsilon &- 2n(t_{n}-s_{n}) + \psi(q_{n}) -  \int_{\mathbb{R}_{+}^{*}} \mathds{1}_{\left\{ q_n + z \in \mathcal Q \right\}} zH^{b} \left( \frac{u(t_{n},q_{n}) -  u(t_{n},q_{n}+z) }{z}\right) \mu^{b}(dz)\\
&-  \int_{\mathbb{R}_{+}^{*}} \mathds{1}_{\left\{ q_n - z \in \mathcal Q \right\}} z H^{a} \left(\frac{u(t_{n},q_{n}) - u(t_{n},q_{n}-z)}{z} \right) \mu^{a}(dz) -  \mathcal H \left( 2n \left( q_n - y_n \right), q_n  \right) \leq 0. \nonumber
\end{split}
\end{equation}

Furthermore, $(s_{n},y_{n})$ is a minimum point of $v-\xi_{n}$ where $$\xi_{n}(s,y) = u(t_{n},y_{n}) - n(q_{n}-y)^{2} - n(t_{n}-s)^{2} - \varepsilon(2T-t_{n}-s) ,$$ and by the same argument
\begin{equation}
\begin{split}
-&\varepsilon - 2n(t_{n}-s_{n}) + \psi(y_{n}) -  \int_{\mathbb{R}_{+}^{*}}\mathds{1}_{\left\{ y_n + z \in \mathcal Q \right\}} zH^{b} \left( \frac{v(s_{n},y_{n}) -  v(s_{n},y_{n}+z) }{z }\right) \mu^{b}(dz)\\
&- \int_{\mathbb{R}_{+}^{*}} \mathds{1}_{\left\{ y_n - z \in \mathcal Q \right\}}z H^{a} \left(\frac{v(s_{n},y_{n}) - v(s_{n},y_{n}-z)}{z} \right) \mu^{a}(dz) -  \mathcal H \left( 2n \left(q_n - y_n \right) , y_n \right) \geq 0. \nonumber
\end{split}
\end{equation}

Therefore by combining the two inequalities we get
\begin{equation}
\begin{split}
& \int_{\mathbb{R}_{+}^{*}} z \Bigg( \mathds{1}_{\left\{ y_n + z \in \mathcal Q \right\}}  H^{b} \bigg( \frac{v(s_{n},y_{n}) -  v(s_{n},y_{n}+z) }{z}\bigg)\\
&\qquad \qquad \qquad \qquad \qquad \qquad- \mathds{1}_{\left\{ q_n + z \in \mathcal Q \right\}} H^{b} \bigg( \frac{u(t_{n},q_{n}) -  u(t_{n},q_{n}+z) }{z}\bigg) \Bigg) \mu^{b}(dz)\\
& + \int_{\mathbb{R}_{+}^{*}} z \Bigg( \mathds{1}_{\left\{ y_n - z \in \mathcal Q \right\}}H^{a} \bigg(\frac{v(s_{n},y_{n}) - v(s_{n},y_{n}-z)}{z} \bigg)\\
&\qquad \qquad \qquad \qquad \qquad \qquad- \mathds{1}_{\left\{ q_n - z \in \mathcal Q \right\}}H^{a} \bigg(\frac{u(t_{n},q_{n}) - u(t_{n},q_{n}-z)}{z} \bigg) \Bigg) \mu^{a}(dz)\\
& \qquad \qquad \qquad  \qquad \qquad  \qquad \qquad \leq -2\varepsilon + \left(\psi(q_{n})- \psi(y_{n}) \right) - \left( \mathcal H \left( 2n \left(q_n - y_n \right) , q_n \right) - \mathcal H \left( 2n \left(q_n - y_n \right) , y_n \right) \right). \nonumber
\end{split}
\end{equation}
By rearranging the terms, we get:
\begin{equation}
\begin{split}
\label{bigineq}
&  \int_{\mathbb{R}_{+}^{*}} \mathds{1}_{\left\{ y_n + z \in \mathcal Q \right\} \cap \left\{ q_n + z \in \mathcal Q \right\} } z \Bigg( H^{b} \bigg( \frac{v(s_{n},y_{n}) -  v(s_{n},y_{n}+z) }{z}\bigg)\\
&- H^{b} \bigg( \frac{u(t_{n},q_{n}) -  u(t_{n},q_{n}+z) }{z}\bigg) \Bigg) \mu^{b}(dz)\\
& +  \int_{\mathbb{R}_{+}^{*}} \mathds{1}_{\left\{ y_n - z \in \mathcal Q \right\} \cap \left\{ q_n - z \in \mathcal Q \right\} } z \Bigg( H^{a} \bigg(\frac{v(s_{n},y_{n}) - v(s_{n},y_{n}-z)}{z} \bigg)\\
&- H^{a} \bigg(\frac{u(t_{n},q_{n}) - u(t_{n},q_{n}-z)}{z} \bigg) \Bigg) \mu^{a}(dz)\\
& \leq-2\varepsilon + \left(\psi(q_{n})- \psi(y_{n}) \right) -  \left( \mathcal H \left( 2n \left(q_n - y_n \right) , q_n \right) - \mathcal H \left( 2n \left(q_n - y_n \right) , y_n \right) \right)\\
&-  \int_{\mathbb{R}_{+}^{*}} \mathds{1}_{\left\{ y_n + z \in \mathcal Q \right\} \cap \left\{ q_n + z \not\in \mathcal Q \right\} } z  H^{b} \bigg( \frac{v(s_{n},y_{n}) -  v(s_{n},y_{n}+z) }{z}\bigg) \mu^{b}(dz)\\
&+  \int_{\mathbb{R}_{+}^{*}} \mathds{1}_{\left\{ y_n + z \not\in \mathcal Q \right\} \cap \left\{ q_n + z \in \mathcal Q \right\} } z  H^{b} \bigg( \frac{u(t_{n},q_{n}) -  u(t_{n},q_{n}+z) }{z}\bigg) \mu^{b}(dz)\\
&- \int_{\mathbb{R}_{+}^{*}} \mathds{1}_{\left\{ y_n - z \in \mathcal Q \right\} \cap \left\{ q_n-z\not\in \mathcal Q \right\} } z  H^{a} \bigg(\frac{v(s_{n},y_{n}) - v(s_{n},y_{n}-z)}{z} \bigg) \mu^{a}(dz)\\
&+  \int_{\mathbb{R}_{+}^{*}} \mathds{1}_{\left\{ y_n - z\not\in \mathcal Q \right\} \cap \left\{ q_n - z\in \mathcal Q \right\} } z  H^{a} \bigg(\frac{u(t_{n},q_{n}) - u(t_{n},q_{n}-z)}{z} \bigg)\mu^{a}(dz).
\end{split}
\end{equation}

We know that $q_n,y_n \underset{n \rightarrow +\infty}{\longrightarrow} \bar{q}$. Therefore, $
\left(\psi(q_{n})- \psi(y_{n}) \right)\xrightarrow[n \rightarrow +\infty]{} 0.$\\

Moreover, by Lemma~\ref{lemmHronde}, there exists a constant $C_{\mathcal H}>0$ such that for all $n$, 
$$\left| \mathcal H \left( 2n \left(q_n - y_n \right) , q_n \right) - \mathcal H \left( 2n \left(q_n - y_n \right) , y_n \right) \right| \le C_{\mathcal H} \left(1+ 2n | q_n - y_n|\right)|q_n-y_n|\xrightarrow[n \rightarrow +\infty]{} 0.$$

We also have for almost every $z>0$ that $\mathds{1}_{\left\{ y_n + z \in \mathcal Q \right\} \cap \left\{ q_n + z \not\in \mathcal Q \right\} }\xrightarrow[n \rightarrow +\infty]{} 0.$\\

By Lemma~\ref{zHbounded}, the term $z  H^{b} \bigg( \frac{v(s_{n},y_{n}) -  v(s_{n},y_{n}+z) }{z}\bigg)$ is bounded uniformly in $n$ and $z,$ and by the absolute continuity of $\mu^{b}$, the dominated convergence theorem enables us to conclude that:
$$\int_{\mathbb{R}_{+}^{*}} \mathds{1}_{\left\{ y_n + z \in \mathcal Q \right\} \cap \left\{ q_n + z \not\in \mathcal Q \right\} } z  H^{b} \bigg( \frac{v(s_{n},y_{n}) -  v(s_{n},y_{n}+z) }{z}\bigg) \mu^{b}(dz)\xrightarrow[n \rightarrow +\infty]{} 0.$$
By the same reasoning:
$$ \int_{\mathbb{R}_{+}^{*}} \mathds{1}_{\left\{ y_n + z \not\in \mathcal Q \right\} \cap \left\{ q_n + z\in \mathcal Q \right\} } z  H^{b} \bigg( \frac{u(t_{n},q_{n}) -  u(t_{n},q_{n}+z) }{z}\bigg) \mu^{b}(dz)\xrightarrow[n \rightarrow +\infty]{} 0,$$
$$ \int_{\mathbb{R}_{+}^{*}} \mathds{1}_{\left\{ y_n - z \in \mathcal Q \right\} \cap \left\{ q_n-z \not\in \mathcal Q \right\} } z  H^{a} \bigg(\frac{v(s_{n},y_{n}) - v(s_{n},y_{n}-z)}{z} \bigg) \mu^{a}(dz)\xrightarrow[n \rightarrow +\infty]{} 0,$$
$$ \int_{\mathbb{R}_{+}^{*}} \mathds{1}_{\left\{ y_n - z \not\in \mathcal Q \right\} \cap \left\{ q_n - z \in \mathcal Q \right\} } z  H^{a} \bigg(\frac{u(t_{n},q_{n}) - u(t_{n},q_{n}-z)}{z} \bigg)\mu^{a}(dz)\xrightarrow[n \rightarrow +\infty]{} 0.$$
We can thus choose $n$ large enough so that the right-hand side of \eqref{bigineq} is negative.\\

However, on the left-hand side of \eqref{bigineq}, all the integrals are always nonnegative; indeed, we have 
\begin{equation*}
\begin{split}
& u(t_{n},q_{n}-z) - v(s_{n},y_{n}-z) - n (q_{n}-y_{n})^{2} - n(t_{n}-s_{n})^{2}- \varepsilon (2T-t_{n}-s_{n}) \\
\le & \  u(t_{n},q_{n}) - v(s_{n},y_{n}) - n (q_{n}-y_{n})^{2} - n(t_{n}-s_{n})^{2}- \varepsilon (2T-t_{n}-s_{n}) ,
\end{split}
\end{equation*}
therefore $v(s_{n},y_{n})-v(s_{n},y_{n}-z) \le u(t_{n},q_{n})-u(t_{n},q_{n}-z) $ and as $H^{a}$ is nonincreasing, we get the result (the proof is identical for the integrals with $H^{b}$).\\

Therefore, the left-hand side is nonnegative for every $n$. But, as we said before, for $n$ large enough, the right-hand side of \eqref{bigineq} is negative, which yields a contradiction.\\

In conclusion, we necessarily have $\underset{[0,T] \times \mathcal Q}{\sup}  (u-v) \leq 0$.\newline

\end{proof}

\begin{thm}
$\theta$ is a continuous function. Morever, $\theta$ is the only viscosity solution to \eqref{eqn:HJB}.\\
\end{thm}

\begin{proof}
We know that $\theta$ is a bounded viscosity solution of \eqref{eqn:HJB}, and in particular, $\theta_{*}$ is a bounded supersolution of \eqref{eqn:HJB}, $\theta^{*}$ is a bounded subsolution of \eqref{eqn:HJB}, and $\theta_{*}(T,.)=\theta^{*}(T,.)=-\ell$.\\

Hence $\theta_*$ and $\theta^*$ verify the assumptions of Theorem~\ref{unic}, and we get that $\theta_* \geq \theta^*$ on $[0,T] \times \mathcal Q.$ But by definition of $\theta_*$ and $\theta^*,$ we have $\theta_* \leq \theta \leq \theta^*.$ Thus we have $\theta_* = \theta = \theta^*,$ and $\theta$ is continuous.\\

Let us now assume that we have another viscosity solution $\tilde{\theta}$. Using the same reasoning as above, $\tilde{\theta}$ is continuous.\\

Because $\tilde{\theta}$ is a subsolution to \eqref{eqn:HJB} and $\theta$ is a supersolution to \eqref{eqn:HJB}, and as $\tilde{\theta}(T,q) = \theta(T,q) = -\ell(q)\ \forall q \in \mathcal Q$, we know by the comparison principle that $\tilde{\theta} \leq \theta$ on $[0,T] \times \mathcal Q$. But we also have that $\tilde{\theta}$ is a supersolution and $\theta$ is a subsolution, so by the same argument we have $\tilde{\theta} \geq \theta$ and finally $\tilde{\theta} = \theta$ on $[0,T] \times \mathcal Q$. Hence the uniqueness.\\
\end{proof}

Using the theory of discontinuous viscosity solutions, we managed to prove that the value function is continuous and that is the unique viscosity solution to the Hamilton-Jacobi equation \eqref{eqn:HJB}. Further regularity  (semi-convexity, semi-concavity, differentiability, $C^1$ regularity) for the value function is an open problem because the inventory is the sum of point processes and an absolutely continuous process, or equivalently, because \eqref{eqn:HJB} contains both finite difference terms like in models à la Avellaneda-Stoikov and partial derivative terms (in the inventory variable) like in models à la Almgren-Chriss.\\

\section{Numerical results}
\label{numericSec}

\subsection{Context and parameters}
In this section, we apply our model to the FX spot market where dealers have access to a variety of trading venues: inter-dealer broker platforms such as Refinitiv and Electronic Broking Services (EBS) and electronic communication networks (ECN) such as Cboe FX, FXSpotStream and LMAX Exchange.\\

In order to derive realistic parameters, we consider a set of HSBC FX streaming clients trading US Dollar against offshore Chinese Renminbi, USDCNH. This set is sufficiently diverse to provide realistic results but by no means complete to fully represent HSBC FX franchise. In particular, in this work we mainly consider connections sensitive to pricing and do not take into account any cross currency trading which may significantly contribute to risk management.\\

For the reference price at any point in time, many choices can be made. Prices are available on many ECNs but these prices may not necessarily indicate commitment (because of the so-called ``Last Look'' practice -- see Oomen~\cite{oomen2017last} and Cartea et al.~\cite{cartea2019foreign}). Therefore, we chose not to rely on ECN prices but rather on the firm prices available on EBS.\\

Since FX market practitioners are used to reason in basis points\footnote{A basis point is one hundredth of a percent.} (bps) as far as quotes are concerned, we decided to factor out $S_0$ from the parameters $\sigma$, $k$, $\delta$, $L$, $\psi$, and $\ell$ (with of course an adjustment of parameters in accordance for intensities). This boils down to factoring out $S_0$ from the value function~$\theta$ and we can therefore reason in bps throughout (up to a little approximation because the base price at any time $t$ should be $S_t$, not $S_0$, but this makes very little difference given the time frame of the problem).\\

In order to use the model in practice, we need to discretize trade sizes. We take trade sizes corresponding to those of the pricing ladder typically streamed to clients.\footnote{In practice, of course, even if there is a lot of volume for specific sizes like \$1 million, volumes are ``continuously'' distributed. A weighted average (called VWAP) is used in the FX world to set a price for volumes that are not on the grid defining the price ladder broadcast by the dealer. We believe that our discretization does not raise any problem for practical applications.} We therefore discretize the distribution of sizes with 4 possible sizes: $z^1 = 1\ \textrm{M\$}$, $z^2 = 5\ \textrm{M\$}$, $z^3 = 10\ \textrm{M\$}$, and $z^4 = 20\ \textrm{M\$}$.\\

In what follows we consider the following parameters\footnote{After assigning trades to one of the 4 size buckets, we used a standard maximum likelihood procedure on our dataset of trades and quotes to estimate the parameters for $\mu^a$, $\mu^b$, $\Lambda^a$ and $\Lambda^b$. Regarding execution costs and market impact, estimations have been made using an approach similar to that of~\cite{almgren2005direct}.} (rescaled as above):

\begin{itemize}
    \item Volatility parameter: $\sigma = 50\ \textrm{bps} \cdot \textrm{day}^{-\frac{1}{2}}$.
    \item Discretization of $\mu^a$ and $\mu^b$ over the above set of sizes:
     $p^1 = 0.76$, $p^2 = 0.15$, $p^3 = 0.075$ and $p^4 = 0.015$.
    \item Logistic intensity functions $\Lambda^b(\delta) = \Lambda^a(\delta) = \lambda_{R} \frac 1{1 + e^{\alpha_\Lambda + \beta_\Lambda \delta}}$,     with $\lambda_{R} = 1000\ \textrm{day}^{-1}$, $\alpha_\Lambda = -1$, and $\beta_\Lambda = 10\ \textrm{bps}^{-1}$. This would correspond to an average of $\frac{1000}{1+e^{-1}} \simeq 731$ trades per day if the quote was the reference price and a drop to $\frac {1000}{1+e^{1}} \simeq 269 $ trades per day if the quote was worsen by 0.2 bps.
    \item $L : v\in \mathbb R \mapsto \eta v^2 + \phi |v|$ with $\eta = 10^{-5}\ \textrm{bps}\cdot \textrm{day} \cdot \textrm{M\$}^{-1}$ and $\phi = 0.1\ \textrm{bps}$.
    \item Permanent market impact: $k=0.005\ \textrm{bps} \cdot \textrm{M\$}^{-1}$.
    \item $\psi : q\in \mathbb R \mapsto \frac{\gamma}{2} \sigma^2 q^2$ with $\gamma=0.0005\ \textrm{bps}^{-1} \cdot \textrm{M\$}^{-1}$.
    \item $\ell(q) = -\frac k2 q^2.$
    \item Time horizon given by $T = 0.05\ \textrm{day}$. This horizon (although it seems short) ensures convergence of quotes and execution rates towards their stationary values at time $t=0$ (see Figures~\ref{conv_exec} and~\ref{conv_deltas}).
\end{itemize}

We impose risk limits in the sense that no trade that would result in an inventory $|q|  > \tilde q$ is admitted, where $\tilde q = 100\ \textrm{M\$}$. We then approximate the solution $\theta$ to \eqref{eqn:HJB} using a monotone implicit Euler scheme on a grid with 201 points for the inventory.

\subsection{Results}

We plot in Figures~\ref{conv_exec} and~\ref{conv_deltas} the optimal execution rate and optimal bid quotes (ask quotes are symmetric) as a function of time for different values of the inventory.\\

\begin{figure}[!h]\centering
\includegraphics[width=\textwidth]{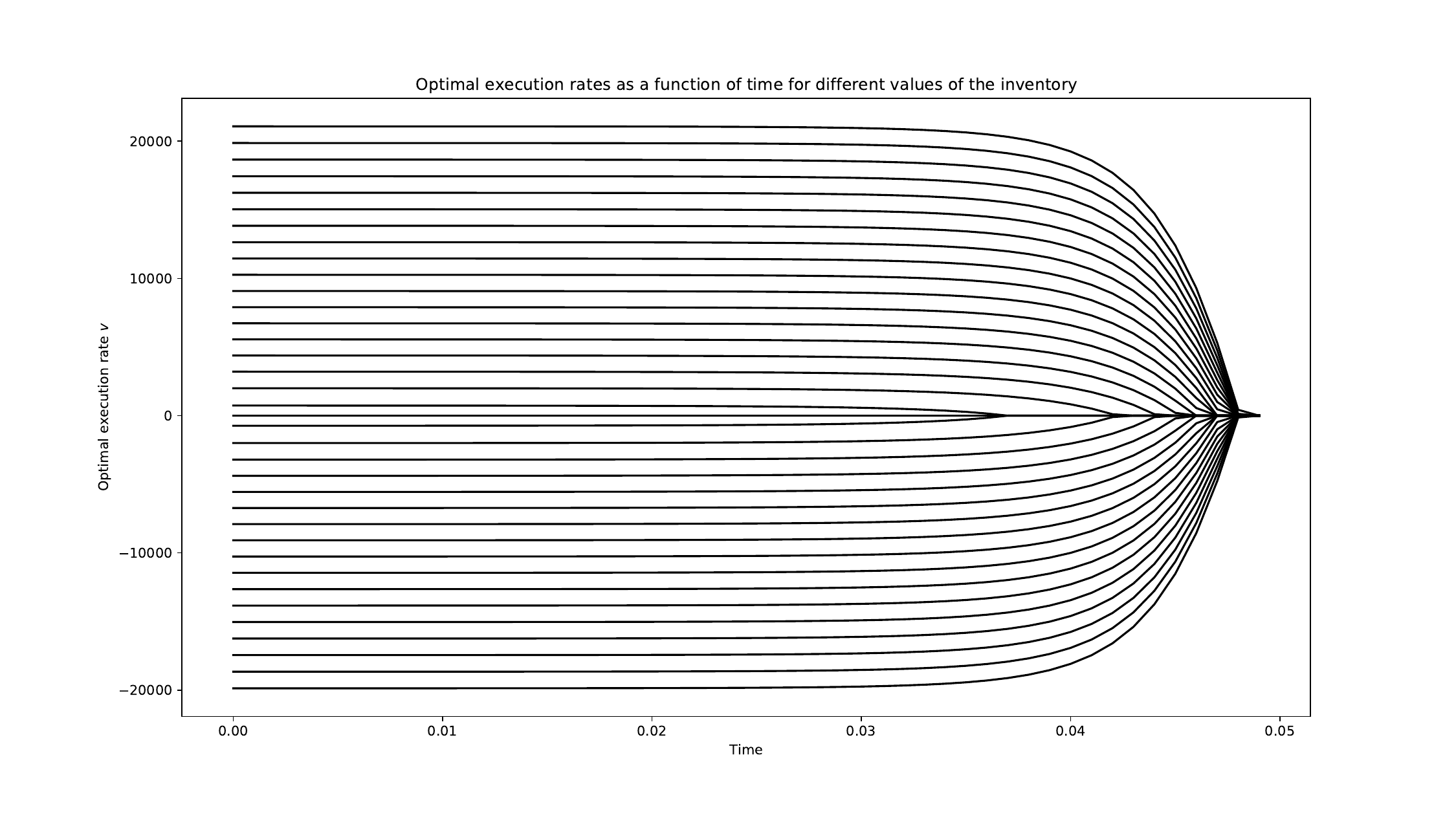}\\
\caption{Optimal execution rates as a function of time for values of the inventory going from $-100$ M\$ to $100$~M\$ by steps of $5$ M\$.}\label{conv_exec}
\end{figure}

\begin{figure}[!h]\centering
\hspace*{-8mm}\includegraphics[width=1.05\textwidth]{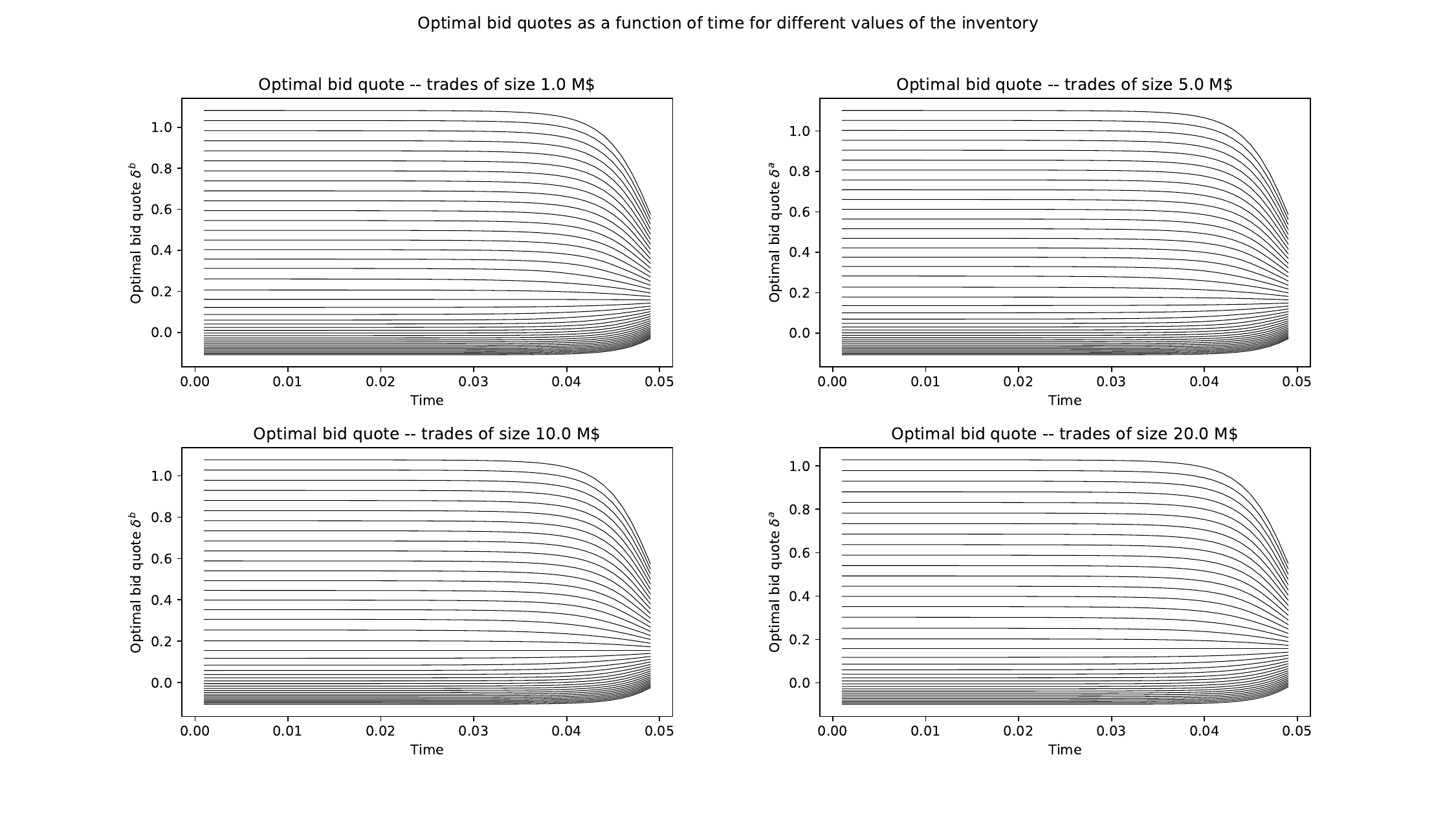}\\
\caption{Optimal bid quotes as a function of time for values of the inventory going from $-100$ M\$ to $100$~M\$ by steps of $5$ M\$.}\label{conv_deltas}
\end{figure}

We observe that the chosen time horizon $T$ is indeed sufficient to reach stationary values. The existence of stationary execution rates and quotes is linked to classical results about the long-time behavior of solutions of Hamilton-Jacobi equations (see for instance~\cite{bardi2008optimal} for the continuous-space case and~\cite{gueant2020optimal} for the discrete-space case). The fast convergence towards these stationary values in our examples is related to the high level of liquidity for the currency pair we consider.\\

The value function (at time $t=0$) is plotted in Figure~\ref{val_func_PA}.\\

\begin{figure}[!h]\centering
\includegraphics[width=0.95\textwidth]{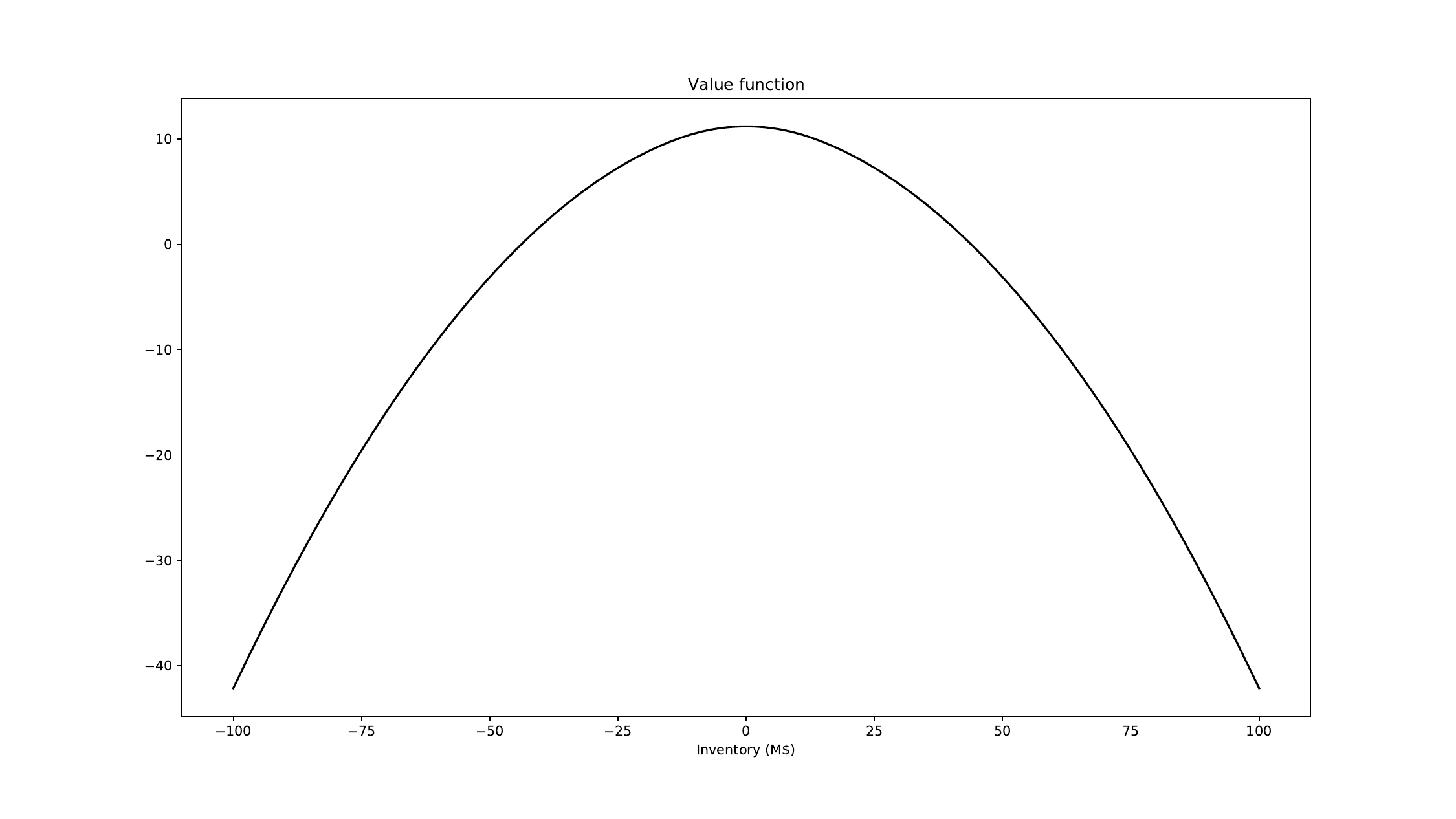}\\
\caption{Value function $q \in \mathcal Q  \mapsto \theta(0,q)$.}\label{val_func_PA}
\end{figure}

We plot in Figure~\ref{opt_rate_PA} the optimal execution rate (from now on we focus on stationary values) as a function of inventory.

\begin{figure}[!h]\centering
\includegraphics[width=0.95\textwidth]{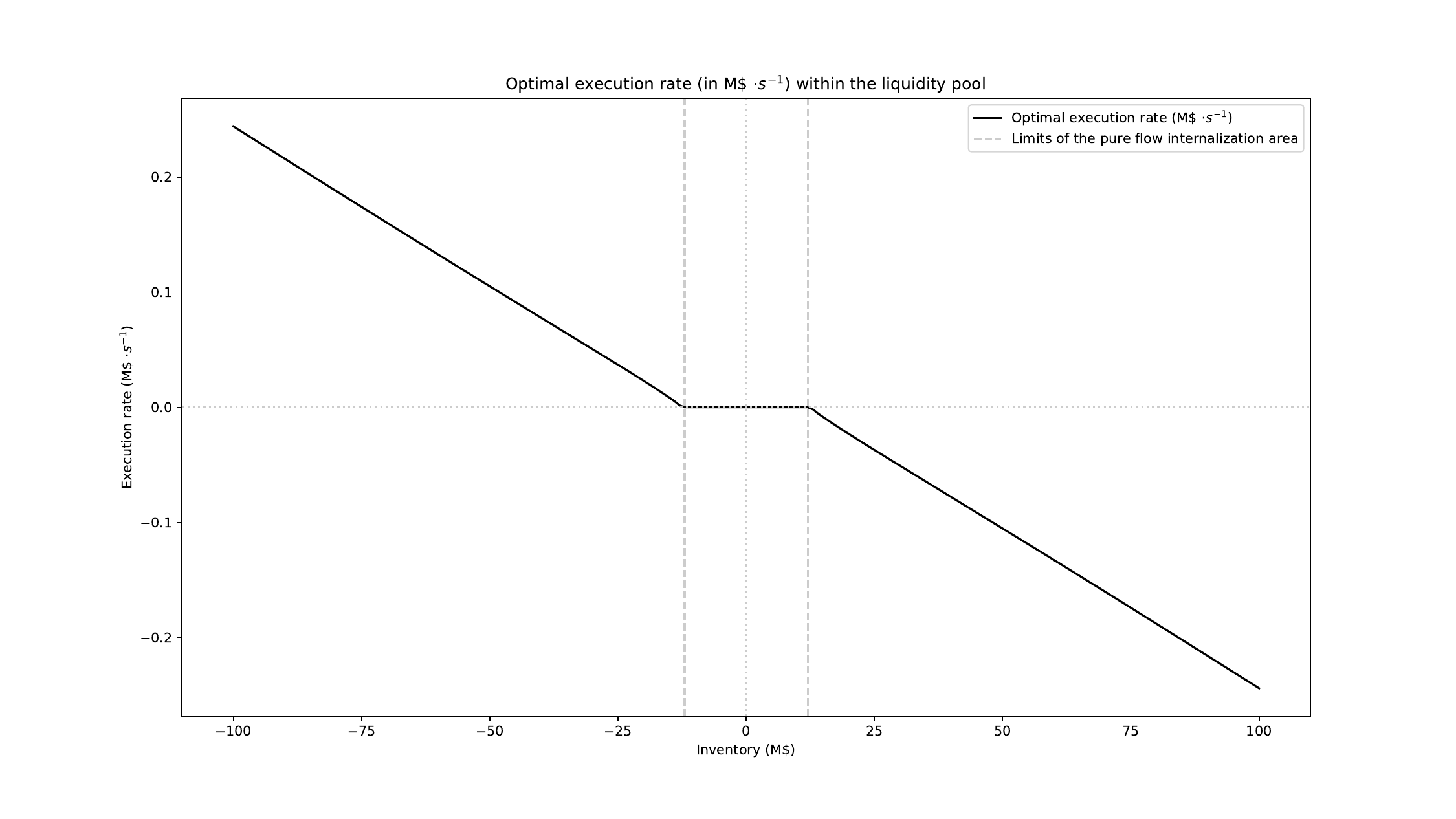}\\
\caption{Optimal execution rate as a function of the inventory.}\label{opt_rate_PA}
\end{figure}

Execution rate is a nonincreasing function of inventory but we observe an interesting effect in Figure~\ref{opt_rate_PA}: a plateau around zero, thereafter referred to as the pure flow internalization area. This is due to the execution costs (especially the proportional transaction costs linked to the parameter $\phi$ in our case) that discourage the dealer to buy or sell externally in the liquidity pool when their inventory is small enough (they prefer to bear this small risk than to pay the execution costs). Permanent impact also discourages external execution. We recall that permanent impact has no influence on classical continuous-time Almgren-Chriss optimal schedules. The reason is that it is proportional to the overall quantity and thus independent of the way the order is executed. The situation is quite different here as no market impact is expected when the flow is internalized. Therefore, external trading brings additional relative cost by pushing the expected price for all the subsequent fills.\\

Coming to optimal quotes, we plot in Figure~\ref{opt_quotes_PA} the four functions 
$$q \mapsto - \bar \delta^b(q, z^k),\ k \in \{1,\ldots,4\},$$
and the four functions
$$q \mapsto \bar \delta^a(q, z^k),\ k \in \{1,\ldots,4\},$$
where $\bar \delta^b$ and $\bar \delta^a$ represent the (stationary) optimal quotes as a function of inventory and size, at the bid and at the ask, respectively.\\

\begin{figure}[!h]\centering
\hspace*{-8mm}\includegraphics[width=1.1\textwidth]{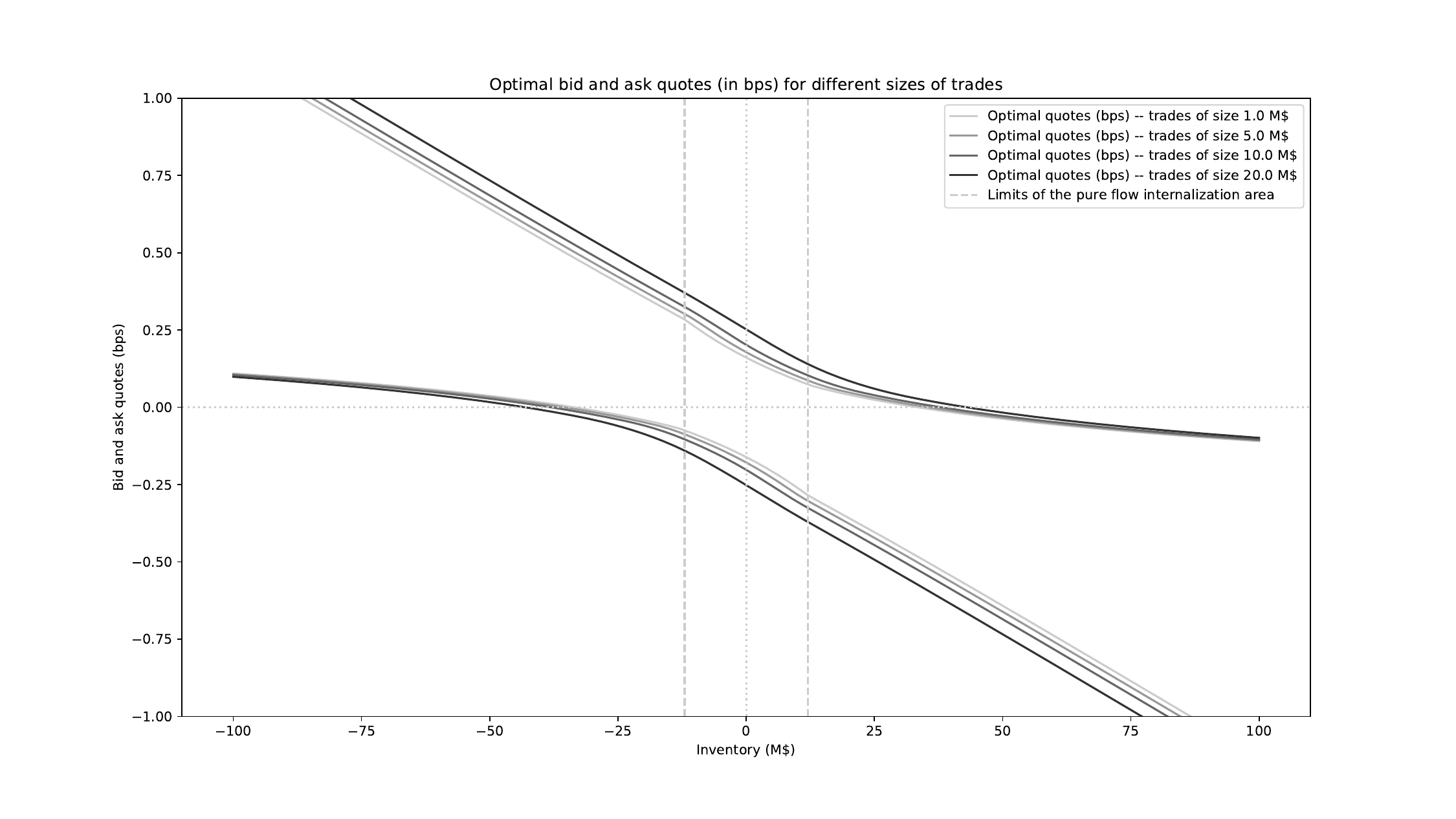}\\
\caption{Optimal bid (bottom) and ask (top) quotes for different trade sizes as a function of the inventory.}\label{opt_quotes_PA}
\end{figure}

We see that accounting for size impacts the optimal bid and ask quotes. The monotonicity of the quotes is of course unsurprising. It is noteworthy that no market spread was parametrically introduced into the model during the estimation of the logistic parameters. Therefore, it is interesting to compare the bid-ask spread we obtain with the actual market spread. Our current estimation produces 0.32 bps for \$1M size. The average composite interbank spread of USDCNH at London open as of this writing (early June 2021) is 0.38 bps.\\

Throughout this section, the optimal quotes are those derived from Lemma~\ref{lemmH} and the optimal execution rates are those derived from Lemma~\ref{lemmHronde}. To confirm empirically that these controls are in line with the value function obtained with our numerical scheme, we performed Monte-Carlo simulations using those controls. The comparison between the value function approximated numerically and the proceed of the Monte-Carlo simulations is plotted in Figure~\ref{MCvf_PA}. We see that the values coincide.\\

\begin{figure}[!h]\centering
\hspace*{-8mm}\includegraphics[width=1.1\textwidth]{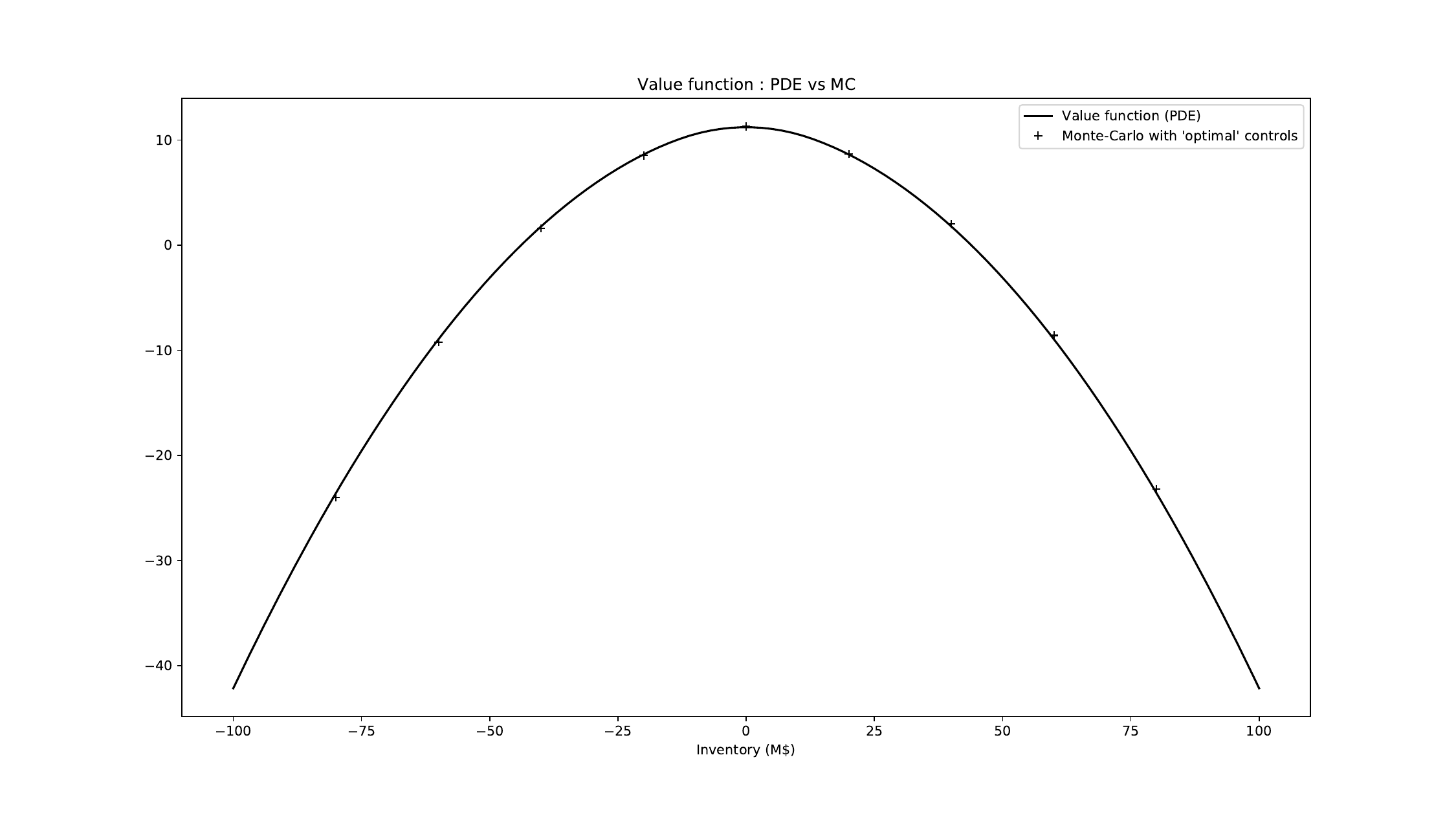}\\
\caption{Value function obtained by playing the optimal quotes and execution rates during a Monte-Carlo simulation compared with the value function computed with an implicit scheme.}\label{MCvf_PA}
\end{figure}

\subsection{Comparative statics regarding the pure flow internalization area}

We now study the influence of the parameters on the width of the pure flow internalization area.\\

We plot in Figure~\ref{opt_rate_psi_PA} the optimal execution rate of the dealer as a function of their inventory, when the execution cost parameter $\phi$ is set to $0.3\ \textrm{bps}$. We see that increasing $\phi$ leads to a wider pure flow internalization area: the dealer is less inclined to pay for immediate hedging and waits for their inventory to reach a higher level of risk to start trading externally in the liquidity pool.\\

We plot in Figure~\ref{opt_rate_k_PA} the optimal execution rate of the dealer as a function of their inventory, when the permanent market impact parameter $k$ is set to $0.01\ \textrm{bps} \cdot \textrm{M\$}^{-1}$. We see that increasing $k$ leads to a wider pure flow internalization area: the dealer is less inclined to impact the market, and waits for their inventory to reach a higher level of risk to start trading externally in the liquidity pool.\\

We plot in Figure~\ref{opt_rate_gam_PA} the optimal execution rate of the dealer as a function of their inventory, when the risk aversion parameter $\gamma$ is set to $0.005\ \textrm{bps}^{-1} \cdot \textrm{M\$}^{-1}$. We see that increasing $\gamma$ leads to a narrower pure flow internalization area: the dealer is more risk averse, and therefore starts externalizing sooner to bring their inventory closer to $0$.\\

\begin{figure}[h!]\centering
\hspace*{-8mm}\includegraphics[width=1.1\textwidth]{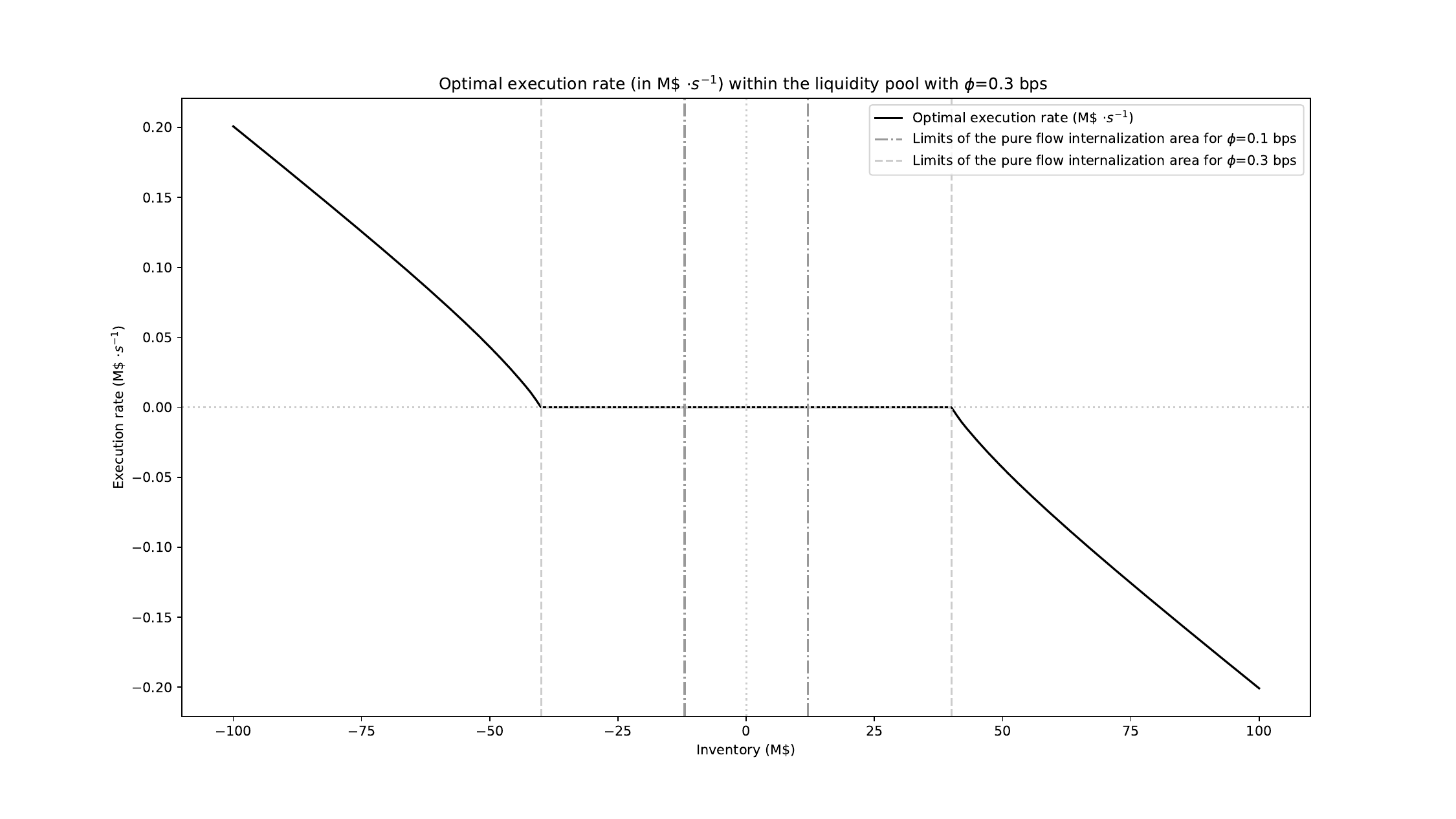}\\
\caption{Optimal execution rate as a function of the inventory when $\phi$ increases.}\label{opt_rate_psi_PA}
\end{figure}

\begin{figure}[h!]\centering
\hspace*{-8mm}\includegraphics[width=1.1\textwidth]{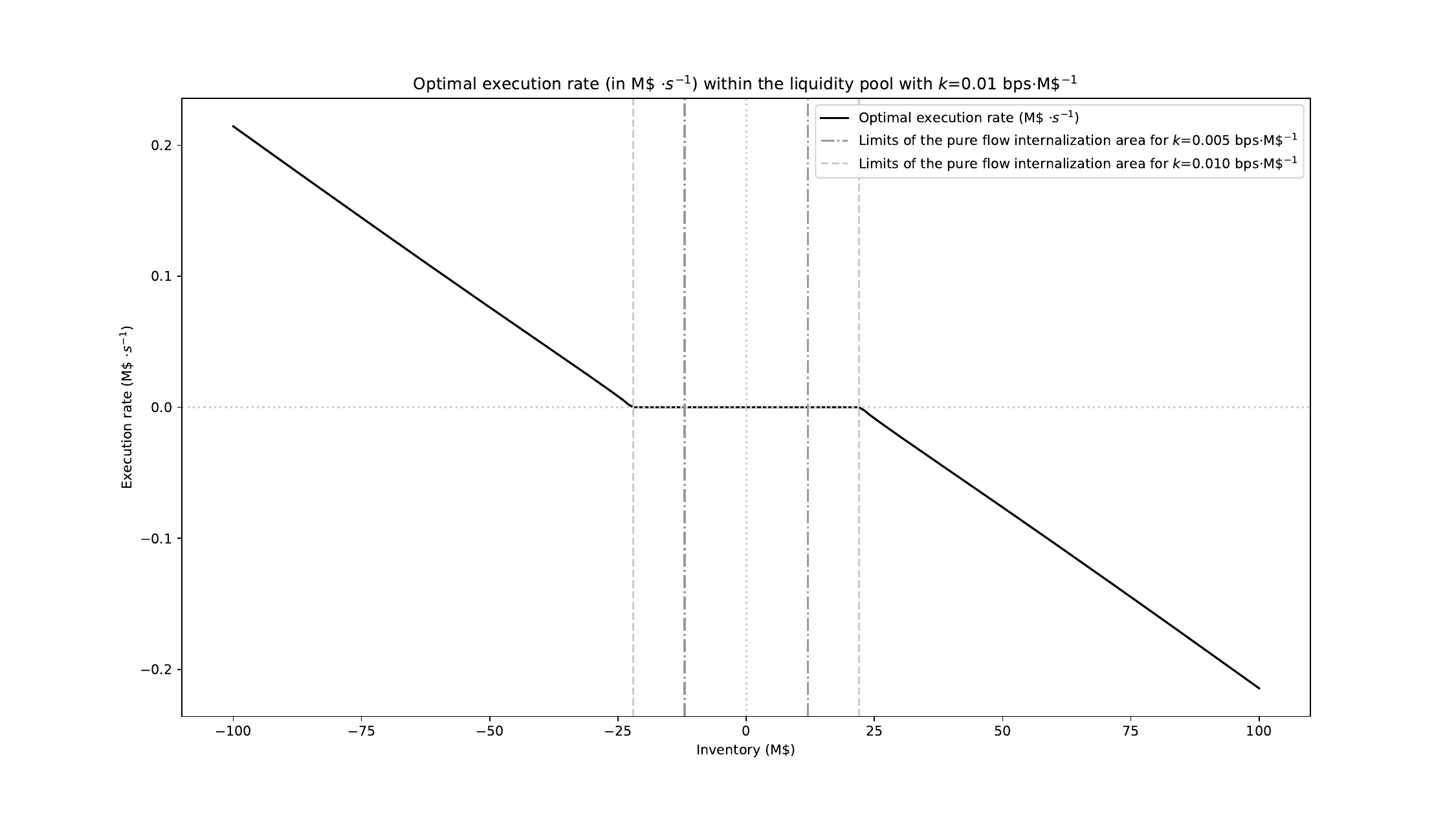}\\
\caption{Optimal execution rate as a function of the inventory when $k$ increases.}\label{opt_rate_k_PA}
\end{figure}

\begin{figure}[h!]\centering
\hspace*{-8mm}\includegraphics[width=1.1\textwidth]{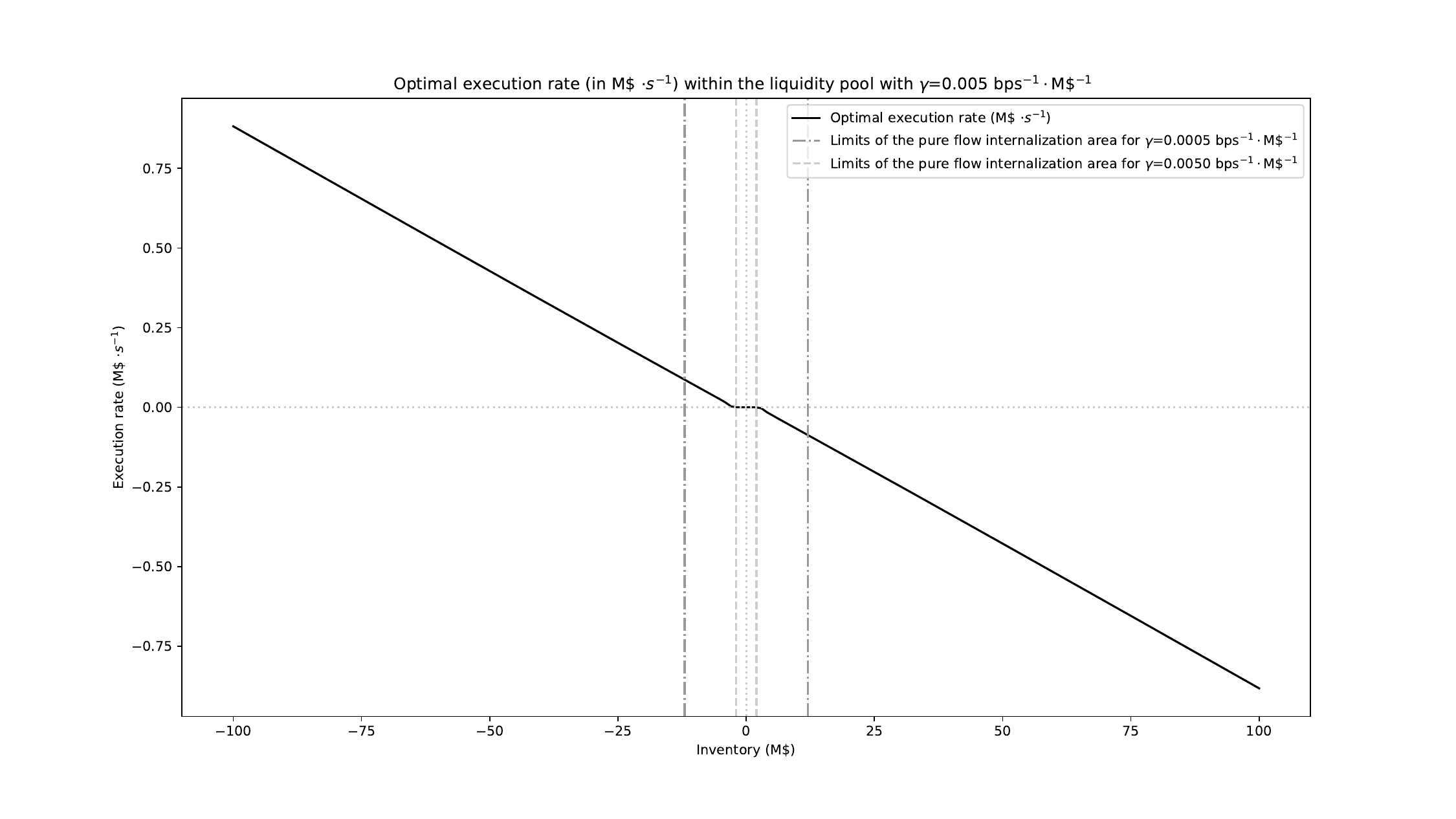}\\
\caption{Optimal execution rate as a function of the inventory when $\gamma$ increases.}\label{opt_rate_gam_PA}
\end{figure}

\begin{figure}[h!]\centering
\hspace*{-8mm}\includegraphics[width=1.1\textwidth]{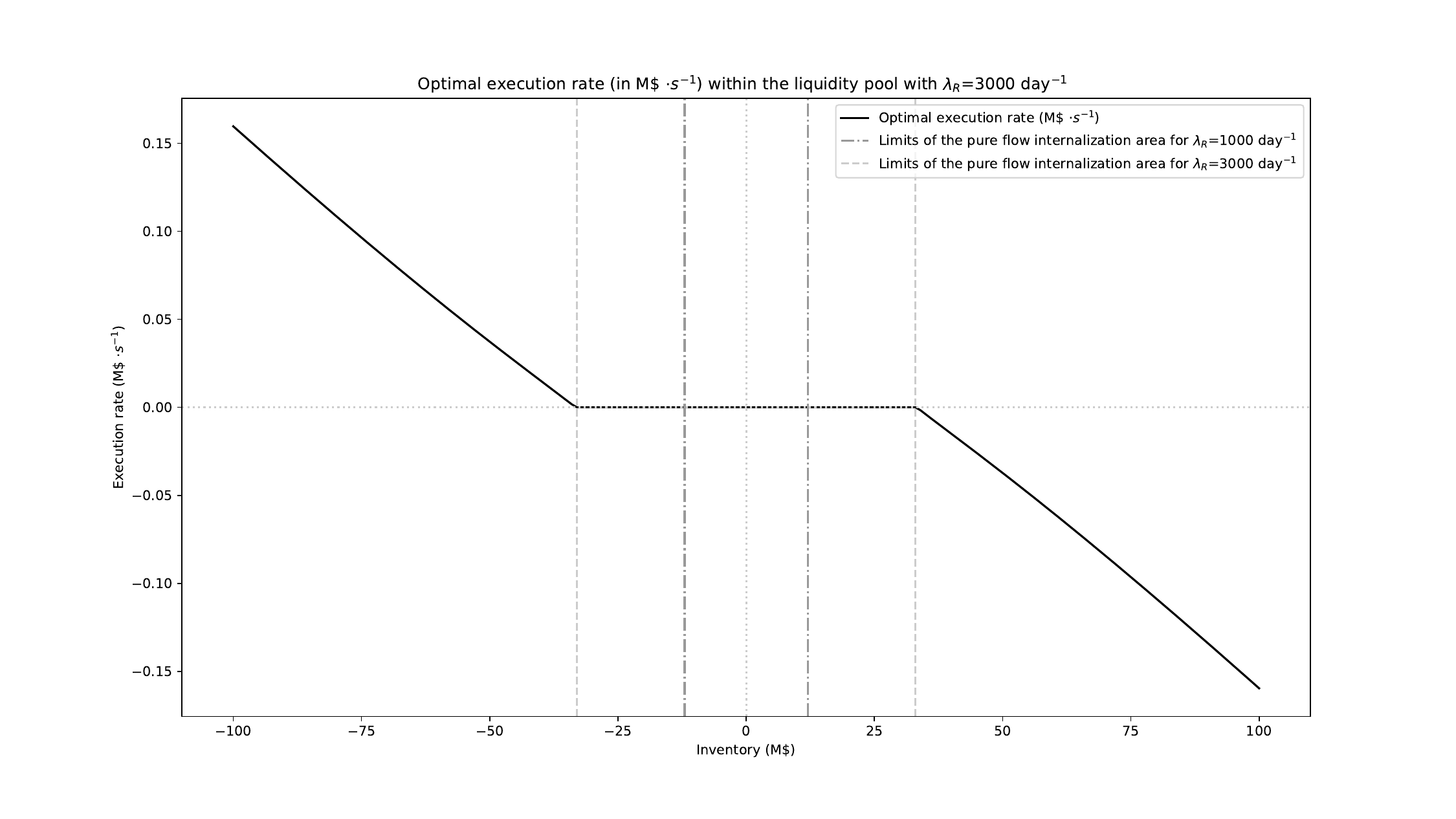}\\
\caption{Optimal execution rate as a function of the inventory when $\lambda_R$ increases.}\label{opt_rate_l_PA}
\end{figure}

We plot in Figure~\ref{opt_rate_l_PA} the optimal execution rate of the dealer as a function of their inventory when the intensity parameter $\lambda_R$ is set to $3000\  \textrm{day}^{-1}$. We see that increasing $\lambda_R$ leads to a wider pure flow internalization area because the dealer has more frequent opportunities to trade and therefore less reasons to pay the costs of externalization.\\

We finally study the impact of asymmetric flows  by considering 
$$\Lambda^b(\delta) = \lambda_{R}^b \frac 1{1 + e^{\alpha_\Lambda + \beta_\Lambda \delta}} \quad \text{and} \quad \Lambda^a(\delta) = \lambda_{R}^a \frac 1{1 + e^{\alpha_\Lambda + \beta_\Lambda \delta}},$$
with $\alpha_\Lambda$ and $\beta_\Lambda$ as before, but
$\lambda_{R}^b = 2500\ \textrm{day}^{-1}$ and $ \lambda_{R}^a = 500\ \textrm{day}^{-1}$ ($5$ times more chance to trade at the bid than at the ask for a similar offset from the reference price). We plot the resulting  execution rate as a function of inventory in Figure~\ref{opt_rate_PA_asymm}. We observe that the pure flow internalization area has widened and shifted to the left. In particular, the dealer tends to sell in the liquidity pool even when their inventory is~$0$ and slightly below $0$ in anticipation of the numerous incoming trades at the bid.\\

\begin{figure}[h!]\centering
\hspace*{-8mm}\includegraphics[width=1.1\textwidth]{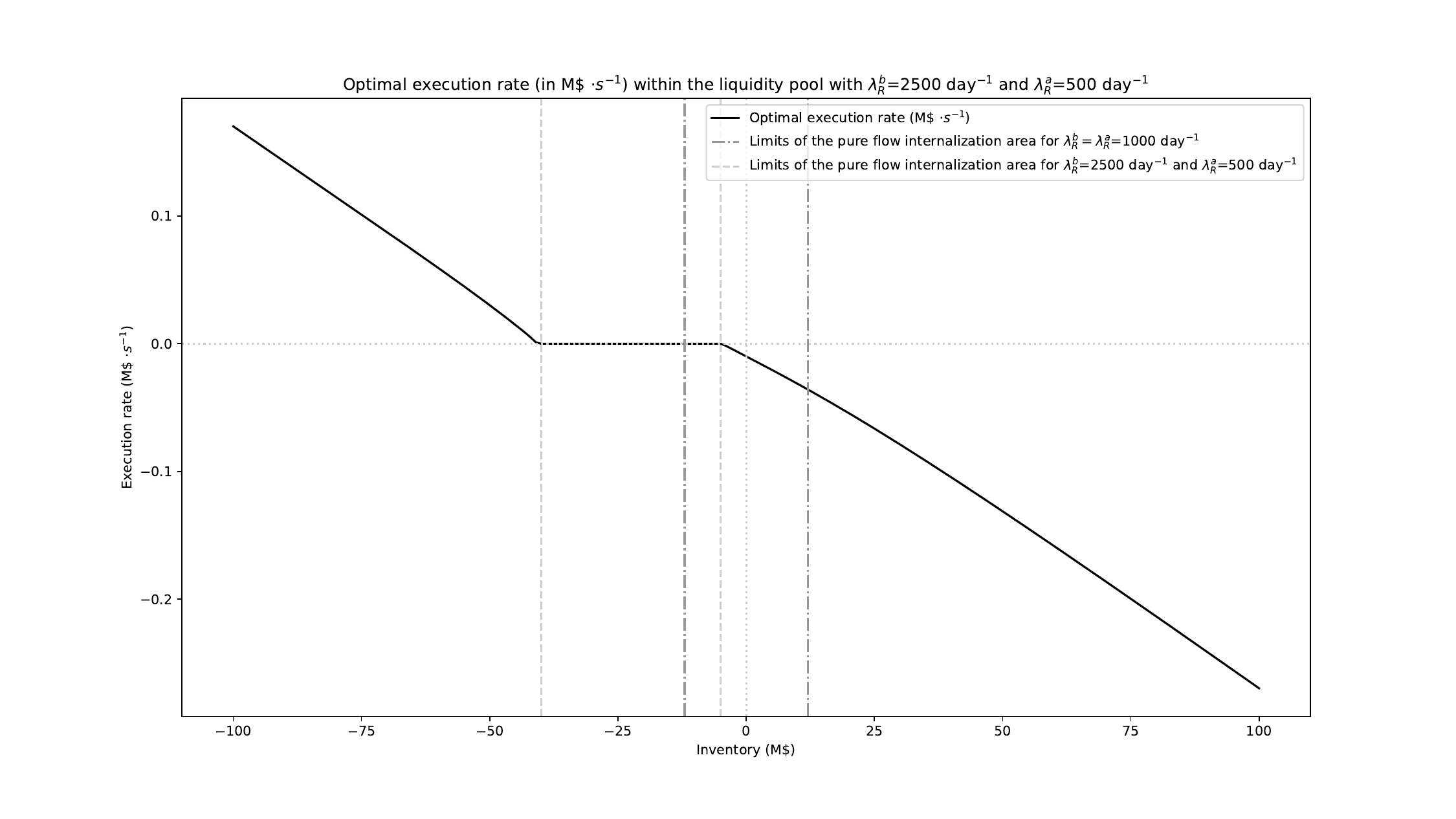}\\
\caption{Optimal execution rate as a function of the inventory for asymmetric flows.}\label{opt_rate_PA_asymm}
\end{figure}

\section*{Conclusion}

In this paper, we generalized the existing market making models to introduce the possibility for dealers to trade in external liquidity pools for hedging purpose. This extension led to a partial integro-differential equation of the Hamilton-Jacobi (HJ) type and we proved that the value function of the problem was its unique continuous viscosity solution. We illustrated our results numerically by solving the equation on a grid using an implicit Euler scheme and computing the optimal quotes and execution rates. We highlighted the existence of a pure flow internalization area. This area depicts a subtle balance between uncertainty, execution cost, and market impact. It is wider for a less risk averse dealer with a larger franchise, exposed to higher transaction costs and market impact.

\section*{Acknowledgment}

The results presented in this paper are part of the research works carried out within the HSBC FX Research Initiative. The views expressed are those of the authors and do not necessarily reflect the views or the practices at HSBC. The authors are grateful to Richard Anthony (HSBC), Bruno Bouchard (Université Paris Dauphine), Nicolas Grandchamp des Raux (HSBC) and Paris Pennesi (HSBC) for helpful discussions and support throughout the project. Two anonymous referees also deserve to be warmly thanked for the comments and suggestions.

\section*{Data Availability Statement}

In order to derive realistic parameters, we considered a set of HSBC FX streaming clients trading the US Dollar against offshore Chinese Renminbi. The set is sufficiently diverse to provide realistic results but by no means complete to fully represent HSBC FX franchise. The data used is not shared.\\

\section*{Appendix}

In this appendix, we prove Proposition~\ref{eqvisco}. More exactly, we only prove the subsolution part (the proof for the supersolution part is identical).\\

Let us first assume that the following inequality holds:

\begin{equation}
\begin{split}
- \frac{\partial \varphi}{\partial t}(\bar{t},\bar{q}) +& \psi(\bar{q}) -  \int_{\mathbb{R}_{+}^{*}} \mathds{1}_{\{\bar{q}+z \in \mathcal{Q}\}}zH^{b} \left(\frac{u(\bar{t},\bar{q}) -  u(\bar{t},\bar{q}+z) }{z}\right) \mu^{b}(dz) \\
&- \int_{\mathbb{R}_{+}^{*}} \mathds{1}_{\{\bar{q}-z \in \mathcal{Q}\}}z H^{a} \left(\frac{u(\bar{t},\bar{q}) - u(\bar{t},\bar{q}-z)}{z} \right) \mu^{a}(dz)  -  \mathcal H\left(\partial_{q} \varphi(\bar{t},\bar{q}) , \bar q\right) \leq 0. \nonumber
\end{split}
\end{equation}

We know that $\forall z >0$:

\begin{equation}
u(\bar{t},\bar{q}-z) - \varphi(\bar{t},\bar{q}-z) \leq u(\bar{t},\bar{q}) - \varphi(\bar{t},\bar{q}). \nonumber
\end{equation}

Thus,

\begin{equation}
H^{a} \left(\frac{u(\bar{t},\bar{q}) - u(\bar{t},\bar{q}-z)}{z} \right) \leq H^{a} \left(\frac{\varphi(\bar{t},\bar{q}) - \varphi(\bar{t},\bar{q}-z)}{z} \right), \nonumber
\end{equation}

and the same holds for $H^{b}$:

\begin{equation}
H^{b} \left(\frac{u(\bar{t},\bar{q}) -  u(\bar{t},\bar{q}+z) }{z}\right) \leq H^{b} \left(\frac{\varphi(\bar{t},\bar{q}) -  \varphi(\bar{t},\bar{q}+z) }{z}\right). \nonumber
\end{equation}

So we get

\begin{equation}
\begin{split}
- \frac{\partial \varphi}{\partial t}(\bar{t},\bar{q}) + &\psi(\bar{q}) -  \int_{\mathbb{R}_{+}^{*}} \mathds{1}_{\{\bar{q}+z \in \mathcal{Q}\}}zH^{b} \left(\frac{\varphi(\bar{t},\bar{q}) -  \varphi(\bar{t},\bar{q}+z) }{z}\right) \mu^{b}(dz)\\
&- \int_{\mathbb{R}_{+}^{*}} \mathds{1}_{\{\bar{q}-z \in \mathcal{Q}\}}z H^{a} \left(\frac{\varphi(\bar{t},\bar{q}) - \varphi(\bar{t},\bar{q}-z)}{z} \right) \mu^{a}(dz) -  \mathcal H\left(\partial_{q} \varphi(\bar{t},\bar{q}), \bar q \right) \leq 0, \nonumber
\end{split}
\end{equation}

and $u$ is a viscosity subsolution.\\

Let us now assume that $u$ is a viscosity subsolution. Without loss of generality, we can assume that $\varphi(\bar{t},\bar{q}) = u(\bar{t},\bar{q})$.\\

Let $B_{\eta}$ be the open ball of center $(\bar{t},\bar{q})$ and radius $\eta >0$. Let $(u_{n})$ be a sequence of smooth functions uniformly (in $n$) bounded such that $u_{n} \geq u\ \forall n$ and $u_{n} \xrightarrow[n\rightarrow +\infty]{} u$ pointwise. Let $\xi$ be a smooth nondecreasing function such that $\xi(x) = 1$ if $x>\eta/4$ and $\xi(x) = 0$ if $x<-\eta/4$. Let $d_{\eta/2}$ be the algebraic distance to $\partial B_{\eta/2}$ (with $d_{\eta/2}>0$ on $B_{\eta/2}$ and $d_{\eta/2}\leq 0$ on $B^{c}_{\eta/2}$); this function is continuously differentiable. We introduce:

\begin{equation}
\varphi^{n}_{\eta} = \varphi \times (\xi\circ d_{\eta/2}) + u_{n} \times (1-\xi\circ d_{\eta/2}). \nonumber
\end{equation}

Then $(\bar{t},\bar{q})$ is still a max of $u - \varphi^{n}_{\eta}$ and $(u - \varphi^{n}_{\eta})(\bar{t},\bar{q})=0$. Furthermore we have $\frac{\partial \varphi^{n}_{\eta}}{\partial t}(\bar{t},\bar{q}) = \frac{\partial \varphi}{\partial t}(\bar{t},\bar{q})$ and $\partial_{q} \varphi^{n}_{\eta}(\bar{t},\bar{q}) = \partial_{q} \varphi(\bar{t},\bar{q}) $. Thus:

\begin{equation}
\begin{split}
- \frac{\partial \varphi}{\partial t}(\bar{t},\bar{q}) + &\psi(\bar{q}) -  \int_{\mathbb{R}_{+}^{*}} \mathds{1}_{\{\bar{q}+z \in \mathcal{Q}\}}zH^{b} \left(\frac{\varphi^{n}_{\eta}(\bar{t},\bar{q}) -  \varphi^{n}_{\eta}(\bar{t},\bar{q}+z) }{z}\right) \mu^{b}(dz)\\
 &-  \int_{\mathbb{R}_{+}^{*}} \mathds{1}_{\{\bar{q}-z \in \mathcal{Q}\}}z H^{a} \left(\frac{\varphi^{n}_{\eta}(\bar{t},\bar{q}) - \varphi^{n}_{\eta}(\bar{t},\bar{q}-z)}{z} \right) \mu^{a}(dz) - \mathcal H\left(\partial_{q} \varphi(\bar{t},\bar{q}) , \bar q\right) \leq 0. \nonumber
\end{split}
\end{equation}

Plus we have $\varphi^{n}_{\eta} \xrightarrow[n\rightarrow +\infty]{}\varphi_{\eta}$ pointwise with $\varphi_{\eta} = \varphi \times (\xi\circ d_{\eta/2}) + u \times (1-\xi\circ d_{\eta/2})$ which is smooth on $B_{\eta/4}$ and such that $\varphi_{\eta} = u$ on $B^{c}_{\eta}$ and $\varphi_{\eta}(\bar{t},\bar{q}) = u(\bar{t},\bar{q})$.\\

By continuity of $H^{a}$ and $H^{b}$, absolute continuity of $\mu^{b}$ and $\mu^{a}$ and by dominated convergence (using the same argument than in Lemma 2 and the fact that the $\varphi^{n}_{\eta}$ are bounded uniformly in $n$) we get:

\begin{equation}
\begin{split}
- \frac{\partial \varphi}{\partial t}(\bar{t},\bar{q}) + &\psi(\bar{q}) -  \int_{\mathbb{R}_{+}^{*}} \mathds{1}_{\{\bar{q}+z \in \mathcal{Q}\}}zH^{b} \left(\frac{\varphi_{\eta}(\bar{t},\bar{q}) -  \varphi_{\eta}(\bar{t},\bar{q}+z) }{z}\right) \mu^{b}(dz) \\
&- \int_{\mathbb{R}_{+}^{*}} \mathds{1}_{\{\bar{q}-z \in \mathcal{Q}\}}z H^{a} \left(\frac{\varphi_{\eta}(\bar{t},\bar{q}) - \varphi_{\eta}(\bar{t},\bar{q}-z)}{z} \right) \mu^{a}(dz) -  \mathcal H\left(\partial_{q} \varphi(\bar{t},\bar{q}), \bar q \right) \leq 0. \nonumber
\end{split}
\end{equation}

By then sending $\eta$ to $0$ and using again dominated convergence, we get the result:

\begin{equation}
\begin{split}
- \frac{\partial \varphi}{\partial t}(\bar{t},\bar{q}) + &\psi(\bar{q}) -  \int_{\mathbb{R}_{+}^{*}} \mathds{1}_{\{\bar{q}+z \in \mathcal{Q}\}}zH^{b} \left(\frac{u(\bar{t},\bar{q}) -  u(\bar{t},\bar{q}+z) }{z}\right) \mu^{b}(dz)\\
& - \int_{\mathbb{R}_{+}^{*}} \mathds{1}_{\{\bar{q}-z \in \mathcal{Q}\}}z H^{a} \left(\frac{u(\bar{t},\bar{q}) - u(\bar{t},\bar{q}-z)}{z} \right) \mu^{a}(dz) -  \mathcal H\left(\partial_{q} \varphi(\bar{t},\bar{q}) , \bar q\right) \leq 0. \nonumber
\end{split}
\end{equation}

\bibliographystyle{plain}

\end{document}